\documentclass[10pt]{article} % , twocolumn]{article}%{proc}%
\usepackage{makeidx} %Work with index.
\usepackage[centertags]{amsmath}

\usepackage{amsfonts}
\usepackage{amssymb}
\usepackage{amsthm}
\usepackage{amsopn}
\usepackage{newlfont}
\usepackage{hhline}

\usepackage{authblk}

\usepackage{rotating}

\usepackage{changepage}

\usepackage{chngcntr}

\usepackage{epsfig, graphicx, lscape} %Work with figures.
\usepackage{delarray}
\usepackage{bm} %Bold in mathematical mode.
\usepackage{enumerate} %Private markings
\usepackage{verbatim} %To be able to set large block of the text as comment
\usepackage[dvips]{color} % Put the text in color

\newlength{\defbaselineskip}
\newcommand{\setlinespacing}[1]%
                                                                          {\setlength{\baselineskip}{#1 \defbaselineskip}}

\theoremstyle{plain}

\newtheorem{example}{Example}

\newtheorem{prop}{Proposition}

\theoremstyle{definition}

\theoremstyle{remark}
\numberwithin{equation}{section}

\begin{document}

%%% \makeatletter\renewcommand{\thechapter}{\@gobble}\makeatother

\counterwithout{equation}{section}
% \counterwithout{example}{section}

% \numberwithout

\title{Affine differential geometry and smoothness maximization as tools for identifying geometric movement primitives}

\author[1]{Felix Polyakov \footnote{Department of Mathematics, Bar Ilan University, Ramat-Gan 5290002, Israel\\ $\; \; $ write.to.felixp@gmail.com}}
%\textsuperscript{1}}
% \affil{Department of Mathematics\\ Bar Ilan University\\ Ramat-Gan 5290002, Israel\\ felix@math.biu.ac.il}

% A class of differential equations for merging movements' kinematic optimality with geometric invariance
\maketitle

% \vspace{-0.4cm}
\begin{abstract}

Neuroscientific studies of drawing-like movements usually analyze neural representation of either geometric (eg. direction, shape) or temporal (eg. speed) features of trajectories rather than trajectory's representation as a whole. This work is about empirically supported mathematical ideas behind splitting and merging geometric and temporal features which characterize biological movements. Movement primitives supposedly facilitate the efficiency of movements' representation in the brain and comply with different criteria for biological movements, among them kinematic smoothness and geometric constraint. Criterion for trajectories' maximal smoothness of arbitrary order $n$ is employed, $n = 3$ is the case of the minimum-jerk model. I derive a class of differential equations obeyed by movement paths for which $n$-th order maximally smooth trajectories have constant rate of accumulating geometric measurement along the drawn path. Constant rate of accumulating equi-affine arc corresponds to compliance with the two-thirds power-law model. Geometric measurement is invariant under a class of geometric transformations and may be chosen to be an arc in certain geometry. Equations' solutions presumably serve as candidates for geometric movement primitives. The derived class of differential equations consists of two parts. The first part is identical for all geometric parameterizations of the path. The second part enforces consistency with desired (geometric) parametrization of curves on solutions of the first part. Equations in different geometries in plane and in space and their known solutions are presented. Connection between geometric invariance, motion smoothness, compositionality and performance of the compromised motor control system is discussed. The derived class of differential equations is a novel tool for discovering candidates for geometric movement primitives.
\end{abstract}

%%%%%%%%%%%%%%%%%%%%%%%%%%%%%%%%%%%%%%%%%%%%%%
%%                                          %%
%% The keywords begin here                  %%
%%                                          %%
%% Put each keyword in separate \kwd{}.     %%
%%                                          %%
%%%%%%%%%%%%%%%%%%%%%%%%%%%%%%%%%%%%%%%%%%%%%%


%%%%%%%%%%%%%%%%%%%%%%%%%%%%%%%%%%%%%%%%%%%%%%
%%                                          %%
%% The Main Body begins here                %%
%%                                          %%
%% Please refer to the instructions for     %%
%% authors on:                              %%
%% http://www.biomedcentral.com/info/authors%%
%% and include the section headings         %%
%% accordingly for your article type.       %%
%%                                          %%
%% See the Results and Discussion section   %%
%% for details on how to create sub-sections%%
%%                                          %%
%% use \cite{...} to cite references        %%
%%  \cite{koon} and                         %%
%%  \cite{oreg,khar,zvai,xjon,schn,pond}    %%
%%  \nocite{smith,marg,hunn,advi,koha,mouse}%%
%%                                          %%
%%%%%%%%%%%%%%%%%%%%%%%%%%%%%%%%%%%%%%%%%%%%%%

%%%%%%%%%%%%%%%%%%%%%%%%% start of article main body
% <put your article body there>

%%%%%%%%%%%%%%%%
%% Background %%
%%
%\section*{Content}
%Text and results for this section, as per the individual journal's instructions for authors. %\cite{koon,oreg,khar,zvai,xjon,schn,pond,smith,marg,hunn,advi,koha,mouse}

\section*{Introduction}

\par Various neuroscientific studies have  analyzed geometric features of primates' and humans' drawing-like movements and their representation in the brain. In particular, single neurons and neural populations in motor cortex were found to be tuned to movement direction \cite{Georgopoulos_etc_1982, Schwartz_1993,  Schwartz:1994, Schwartz_Moran_1999, Moran_Schwartz_reaching_1999, Moran_Schwartz_spiral_1999}. Studies of different types of goal-directed movements, eg. movements to targets, sequential hand movements or movements following prescribed paths indicated that the serial order of submovement\footnote{Meaning serial order of implementing certain geometric entity.}, ``aspects of movement'' (meaning aspects of movement shape and target location) and movement fragments are represented in cortical activity \cite{HochermanS.WiseSP:1991, Averbeck_etc_Georgopoulos1:2003, Averbeck_etc_Georgopoulos2:2003, Hatsopoulos.Xu.Amit:2007, Shanecih.etc.Brown.Williams:2012, Hatsopoulos.Amit:2012}. %
\par Pioneering works by Pollick and Shapiro \cite{Pollick_Shapiro_1997} and Handzel and Flash \cite{Handzel_Flash_1999} reported equivalence of the 2/3 power-law model\footnote{The model establishes relationship between movement's speed and curvature. More details are provided further in text.} to moving with constant equi-affine velocity and proposed relevance of non-Euclidian geometry, and in particular equi-affine, to the mechanisms of biological movements. Later studies revealed presence of features characterizing equi-affine invariants in empirical data recorded during production and perception of biological motion \cite{Polyakov:2001, Polyakov_etc_2001, Polyakov:2006, Levit-Binnun.Schechtman.Flash:2006, Flash_Handzel:2007, Dayan.etal.Flash:2007, Polyakov_et_al_B.Cyb:2009, Polyakov_et_al_PLoS_C_B:2009, Maoz:2009, Pollick:2009, Casile:2010, Maoz:2014, Meirovitch.Harris.Dayan.Arieli.Flash:2015}. Relevance of a larger, affine, invariance was introduced to the field of motor control by Bennequin at at. \cite{Bennequin:2009, Fuchs:2010, Pham.Bennequin:2012}.
%

%
%Primitives may exist at the kinetic, kinematic
%and/or neural levels of the motor system (Mussa-Ivaldi and
%Solla 2004; Flash and Hochner 2005; Giszter et al. 2007).
\par In addition to the continuous representation of movements, the idea of movement compositionality, i.e. representation of complex movements based on a limited ``alphabet'' of primitive submovements, is analyzed in numerous motor control studies. Studies of monkey and human trajectories suggested that movement primitive can be defined on the operational level as a movement entity that cannot be intentionally stopped before its completion once it has been initiated \cite{Polyakov:2006, Sosnik.Shemesh.Abeles:2007, Polyakov_et_al_PLoS_C_B:2009, sosnik.chaim.flash:2015}. Existence of motor primitives was demonstrated at the level of forces produced by muscles operating on the limbs \cite{Bizzi_Mussa_Ivaldi_Giszter_1991, Giszter.Mussa-Ivaldi.Bizzi:1993, Nichols:1994, Kargo.Giszter:2000, Mussa_Ivaldi_Bizzi_2000,  Mussa_Ivaldi_Solla_2000, Giszter.etal:2007, Giszter.etal:2013}, at the level of muscle synergies \cite{Tresch.Saltiel.Bizzi:1999, d.Avella.Saltiel.Bizzi:2003, Hart.Giszter:2004, Ivanenko.Poppele.Lacquaniti:2004, d.Avella.Portone.Fernandez.Bizzi:2006}, at the level of motion kinematics \cite{Morasso.MussaIvaldi:1982, Flash_Henis_1991, Burdet.Milner:1998, Krebs_et_al_1999, Sanger:2000, Rohrer.Hogan:2003, Fishbach.etc.Houk:2005, Flash.Hochner:2005, Rohrer.Hogan:2006}, at the level of units of computation
in the sensorimotor system \cite{van.Zuelen.Gielen.derGon.Denier:1988, Thoroughman_Shadmehr_2000}, and as a vector cross product between a limb-segment position and a velocity or acceleration \cite{Tanaka.Sejnowski:2014}.
It was proposed that movement primitive is an action of a neuromuscular system controlling automatic synergy whose elements produce stereotypical and repeatable results \cite{Woch.Plamondon:2010, Woch.Plamondon.OReilly:2011}. Decomposition of complex movements into primitives was also implemented for octopus arm movements \cite{Zelman.etc.Flash:2013} and wrist movements in human sign language \cite{Endres:2013}. Recent works analyzed neural representation of movements involving corrective submovements in double-step paradigm \cite{Dickey.Amit.Hatsopoulos:2013, Dipietro.Poizner.Krebs:2014} and provided additional indications that seemingly continuous movements might be represented in the brain at certain hierarchical level in discrete manner.
\par Parabolic shapes were suggested as geometric building blocks of complex drawing-like movements based on mathematical modeling and analysis of kinematic and neurophysiological data of behaving monkeys \cite{Polyakov_etc_2001, Polyakov:2006, Polyakov_et_al_B.Cyb:2009, Polyakov_et_al_PLoS_C_B:2009}. It was shown that analyzed spontaneous movement paths can be represented in a compact way by concatenating parabolic-like shapes. Affine transformations applied to parabolic segments result in parabolic segments. Moreover, any parabolic segment can be mapped to another arbitrary parabolic segment by unique affine transformation such that segments' initial (and final) points are matched as described further in text\footnote{In the part related to affine parametrization of curves}. So, provided a direction of motion, a sequence of concatenated parabolic-like shapes can be obtained by applying a unique sequence of affine transformations to a single parabolic template with prescribed starting point and thus \emph{simplifying the representation of complex movements in the brain}. Such representation could mean that geometric movement primitive is a set of transformations endowed with a primitive geometric shape upon which the transformations are applied. Simulated pattern composed of parabolic segments obtained by applying a sequence of affine transformations to a single parabolic segment resembled actual movement path performed by monkey \cite{Polyakov:2006, Polyakov_et_al_B.Cyb:2009}.
\par It was observed by Flash and Hogan that planar hand trajectories are smooth while non-smoothness was quantified by the cost functional called jerk leading to the model named ``minimum-jerk'' \cite{Hogan:1984, Flash_Hogan_1985}. Viviani and Flash compared predictions of the minimum-jerk and the 2/3 power-law models for some geometric shapes. Work by Richardson and Flash analyzed relationship between the 2/3 power-law model and smoothness of arbitrary degrees using approximation to a number of shapes \cite{Richardson_Flash_2003}. Later Polyakov et. al. derived differential equation whose solutions based on parametrization with the equi-affine arc serve candidates for providing identical predictions to the minimum-jerk and the 2/3 power-law models \cite{Polyakov:2001, Polyakov_etc_2001, Polyakov:2006, Polyakov_et_al_B.Cyb:2009}.
\par \par Earlier attempts to propose models for duration of a movement or of its part were based on optimization principles, e.g. \cite{Flash_Hogan_1985, Uno.Suzuki.Kawato:1989, Harris_Wolpert_1998, Todorov_Jordan_1998, Tanaka.Krakauer.Quan:2006}. Bennequin et. al. proposed that movement timing and invariance arise from mixture of Euclidian, equi-affine and affine geometries \cite{Bennequin:2009, Fuchs:2010}; in other words movement duration is proportional to mixture of arcs in different geometries. This theory accounted well for the kinematic and temporal features of a number of repeatable drawing and locomotion movements and suggested for the first time that single movement trajectory can be simultaneously represented in different geometries. It was proposed that the equi-affine geometry was the most dominant, affine geometry second most important during drawing, and Euclidian second most important during locomotion.  
\par More recently scale invariance in the neural representation of handwriting movements was observed \cite{Kadmon.Flash:2014} (superposition of equi-affine transformations and uniform scaling\footnote{Uniform scaling needed to enrich planar equi-affine group of transformations to affine is represented by either of the two matrices: 1) uniform scaling without reflection: $\mbox{const} \cdot \left(
\begin{array}{cc}
        1 & 0 \\
        0 & 1
     \end{array} \right)$; 2) uniform scaling with reflection over $y$ axis: $\mbox{const} \cdot \left(      \begin{array}{cc}
        1 & 0 \\
        0 & -1
     \end{array} \right)$, $\mbox{const} \neq 0$.} constitutes the group of affine transformations). Level of activation in different motor areas (M1, PMd, pre SMA) was found to be related to the level of motion smoothness acquired during learning to coarticulate point-to-point segments into complex smooth trajectories \cite{sosnik.flash.sterkin.hauptmann.karni:2014}.
\par Here I expand existing mathematical tools aimed at finding primitive geometric shapes which are related to invariance-smoothness criteria beyond approximation of shapes and beyond the minimum-jerk criterion simultaneously. Candidates for primitive shapes have exact functional description for smoothness of arbitrary degrees and can be identified with constant rate of accumulating a feasible geometric measurement (e.g. equi-affine arc) along movement's path and its invariance. So trajectories along those primitive shapes would possess a smoothness feature that was observed in biological motion. Classes of such primitive shapes presumably composing more complex trajectory paths are invariant under classes of geometric transformations thus being able to provide compact representation of complex movement paths in the brain. Given ideas of employing multiple geometric arcs for representing movements \cite{Bennequin:2009, Fuchs:2010}, the derived method of identifying geometric primitives is demonstrated here for invariance in different geometries and the reader can further apply this ready to use machinery on his own for geometric measurements not mentioned in this work.
\subsection*{Prerequisites for the mathematical problem from the motor control studies}
\par Trajectory's smoothness criterion was initially defined as minimization of the integrated squared rate of change of acceleration called also movement jerk \cite{Hogan:1984, Flash_Hogan_1985}, namely:
 \begin{equation}\label{minimum_jerk_classical}
    \displaystyle\int\limits_0^T \left\{
        \left[\frac{d^3 x}{dt^3}\right]^2 + \left[\frac{d^3 y}{dt^3}\right]^2\right\}\, dt\, .
 \end{equation}
The information about movement's trajectory can be split into two parts: (1) geometric specification called also \emph{movement path} and (2) temporal specification defined by a function relating each moment of time to the location on the movement path. The temporal specification is fully determined by the speed of motion along the path. In the original works on the \emph{minimum-jerk} model \cite{Hogan:1984, Flash_Hogan_1985} maximally smooth trajectory is constrained by a starting point, a via-point through which the path has to pass and the end point. So the criterion of minimizing the cost functional in \eqref{minimum_jerk_classical} is endowed with point-wise kinematic constraint on optimal trajectory that passes through one or more via-points. Therefore in the original formulation of the model the entire continuous path of the trajectory has to be revealed simultaneously with identifying movement speed\footnote{The $x(t)$ and $y(t)$ components of the trajectories constrained by via-points and minimizing the cost functional \eqref{minimum_jerk_classical} are composed of pieces of 5th order polynomials with respect to time,  the 3rd order derivatives of $x(t)$, $y(t)$ are continuous \cite{Flash_Hogan_1985}. Minimum-jerk trajectories with a single via-point can be  well approximated with parabolic segments \cite{Polyakov:2006, Polyakov_et_al_B.Cyb:2009} and satisfy isochrony principle stating that different movement portions have nearly the same duration independently of their extent \cite{viviani.terzuolo:1982, Bennequin:2009}. Movement durations from the start to the via-point and from the via-point to the end-point are very similar \cite{Polyakov:2006, Polyakov_et_al_B.Cyb:2009}.}. The minimum-jerk model is widely used and mentioned in different motor control studies.
\par According to the \emph{constrained minimum-jerk} model proposed by Todorov and Jordan \cite{Todorov_Jordan_1998} hand movements tend to maximize the smoothness of drawing or, in other words, minimize the jerk cost
\begin{equation}\label{minimum_jerk_constrained_Todorov}
    \displaystyle\int\limits_0^T \left\{
        \left[\frac{d ^3 x(\sigma_{eu}(t))}{dt^3}\right]^2 + \left[\frac{d ^3 y(\sigma_{eu}(t))}{dt^3}\right]^2 + \left[\frac{d ^3 z(\sigma_{eu}(t))}{dt^3}\right]^2\right\}\, dt
 \end{equation}
 for the prescribed trajectory path $\{x(\sigma_{eu}), y(\sigma_{eu}), z(\sigma_{eu})\}$. That is movement path is already provided as an input to the optimization procedure and only the speed profile has to be found to solve the optimization problem. Executed 3-dimensional curve in the cost functional \eqref{minimum_jerk_constrained_Todorov} is parameterized with Euclidian arc-length
\begin{equation}\label{Euclidian_speed}
    \sigma_{eu}(t) = \int_{0}^{t}\sqrt{\dot{x}(\tau)^2 + \dot{y}(\tau)^2 + \dot{z}(\tau)^2 }d\tau
\end{equation}
with dot denoting differentiation with respect to time $t$.
\begin{example}\nonumber
Trajectory $\mathbf{r}(t) = \mathbf{r}(\sigma_{eu}(t)) = [x(\sigma_{eu}(t)),\, y(\sigma_{eu}(t)),\, z(\sigma_{eu}(t))]$ is fully determined by geometric (not involving time) parametrization\footnote{Geometric parametrization with an arc invariant in certain geometry is called natural parametrization; length $\sigma_{eu}$ provides natural parametrization in Euclidian geometry.} $\mathbf{r}(\sigma_{eu})$ and temporal parametrization of the geometric parameter $\sigma_{eu}(t)$ (or equivalently non-negative speed $\dot{\sigma}_{eu}$). Here geometric parameter $\sigma_{eu}$ is length which is continuously mapped onto a curve described by 3-dimensional differentiable vector function.  \;\;\; $\Box$
\end{example}
%
% \vspace{0.7 cm}
\par The \emph{2/3 power-law} proposed by Lacquaniti, Terzuolo and Viviani \cite{Lacquaniti_Terzuolo_Viviani_1983} establishes a local kinematic constraint on movements. It describes a relationship between geometric properties of movement path and speed of motion along that path, namely
$$K = \mbox{Speed} \cdot \mbox{Curvature}^{1/3} = \mbox{Angular speed} \cdot \mbox{Curvature}^{-2/3} \, , $$
where $K$ is piece-wise constant, speed and curvature are Euclidian.
Empirical observations of the 2/3 power-law model were interpreted as evidence for movement segmentation \cite{Lacquaniti_Terzuolo_Viviani_1983}. The 2/3 power law was also demonstrated in visual perception \cite{Viviani_Stucchi_1992, Levit-Binnun.Schechtman.Flash:2006, Dayan.etal.Flash:2007, Casile:2010, Meirovitch.Harris.Dayan.Arieli.Flash:2015} and locomotion \cite{Vieilledent.Kerlirzin.Dalbera.Berthoz:2001, Ivanenko.Grasso.Macellari.Lacquaniti:2002} studies. Segmentation of hand movements based on powers of trajectory curvature has recently been analyzed in \cite{Endres:2013}.
\par Euclidian speed $\dot{\sigma}_{eu}$ minimizing the cost functional with an arbitrary order of smoothness $n$
 \begin{equation}\label{cost_function_Richardson}
    \displaystyle\int\limits_0^T \left\{
        \left[\frac{d ^n x(t)}{dt^n}\right]^2 + \left[\frac{d ^n y(t)}{dt^n}\right]^2\right\}\, dt
 \end{equation}
was compared to the experimental data for planar point-to-point movements\footnote{An assumption behind some kinematic models for point-to-point movements in plane is that the movement path is a straight line, speed and acceleration of motion at the start and end points of the trajectory are zero.} \cite{Richardson_Flash_2003}. Approximated predictions of movement  speed which minimizes the cost functional with arbitrary order $n$ \eqref{cost_function_Richardson} along a number of periodic paths were derived and compared to the predictions of the 2/3 power-law and experimental data \cite{Richardson_Flash_2003}.
\vspace{0.3cm}
\par From now on differentiation with respect to parameter $\sigma$ \footnote{Usually denotes here geometric measurement along a path.} is denoted with primes and numbers in brackets while differentiation with respect to $t$ \footnote{Denotes here time or an arbitrary parameter.} up to order 3 is denoted with dots:
\begin{eqnarray*}
    f'(\sigma) & \equiv & \frac{df}{d\sigma}, \, f''(\sigma) \equiv \frac{d^2f}{d\sigma^2},\, f'''(\sigma) \equiv \frac{d^3f}{d\sigma^3},\, f^{(k)}(\sigma) \equiv \frac{d^kf}{d\sigma^k} \\
    \dot{f}(\sigma(t)) & \equiv & \frac{df}{dt},\, \ddot{f}(\sigma(t)) \equiv \frac{d^2f}{dt^2},\, \dddot{f}(\sigma(t)) \equiv \frac{d^3f}{dt^3}\;.
\end{eqnarray*}
The notation will be reminded further in text.
\vspace{0.3cm}
\par The 2/3 power-law model
%was reformulated in the language of equi-affine geometry, namely the model
is equivalent to the statement that the equi-affine velocity\footnote{All necessary formulae from affine differential geometry are provided further in text. A general background for the notions of equi-affine geometry which are used in this work can be found elsewhere, eg. \cite{Shirokovy_1959, Guggenheimer:1977}. The book \cite{Shirokovy_1959} provides the most comprehensive treatise on affine differential geometry that I am aware of. English translation
 of relevant parts of \cite{Shirokovy_1959} can be obtained from the author of the
 manuscript (FP) for non-commercial use in research and teaching. Some parts of \cite{Shirokovy_1959} are translated into English in Appendix A of \cite{Polyakov:2006}. } \eqref{eq:ea_speed} of drawing movements is piece-wise constant \cite{Pollick_Shapiro_1997, Handzel_Flash_1999, Polyakov:2001, Polyakov_etc_2001, Polyakov:2006, Flash_Handzel:2007, Polyakov_et_al_B.Cyb:2009}: $\dot{\sigma}_{ea} = \mathrm{const}$.
%Such reformulation enabled researchers to use notions of differential geometry in analysis of biological motion. In particular, one of the novelties in those works was parametrization of the drawn curves with geometric arc relevant for biological movements and different from Euclidian length, specifically with the equi-affine arc $\sigma_{ea}$ \eqref{eq:ea_arc} which is invariant under equi-affine geometric transformations \eqref{eq:equi-affine_transformation}.
Based on the results of empirical studies related to interpretation of the 2/3 power-law in terms of differential geometry and their extension, equi-affine and affine arcs have become relevant parametrization in analysis of biological movements and their neural representation \cite{Pollick_Shapiro_1997, Handzel_Flash_1999, Polyakov:2001, Polyakov_etc_2001, Polyakov:2006, Dayan.etal.Flash:2007, Flash_Handzel:2007, Polyakov_et_al_B.Cyb:2009, Polyakov_et_al_PLoS_C_B:2009, Bennequin:2009, Casile:2010, Fuchs:2010}.
\par Equivalence of the 2/3 power-law to constancy of the planar equi-affine velocity was extrapolated by Pollick et al. \cite{Pollick:2009} to the 1/6 power-law equivalent to the conservation of the spatial equi-affine velocity \eqref{eq:ea3_speed}. Empirical validity of the 1/6 power-law for both action and perception was verified in the works by Maoz, Pollick and others \cite{Maoz:2007, Pollick:2009, Maoz:2009, Maoz:2014}.
\par In case of the constrained minimum-jerk model ($n = 3$ in \eqref{cost_function_Richardson}), the problem of finding the paths whose maximally smooth trajectories satisfy the 2/3 power-law model was studied using parametrization of a path with equi-affine arc \cite{Polyakov:2001, Polyakov_etc_2001, Polyakov:2006, Polyakov_et_al_B.Cyb:2009}. Correspondingly the problem was formulated as finding paths whose maximally smooth trajectories have constant equi-affine velocity and reduced to the necessary condition formulated as differential equation:
\begin{equation}\label{equation_equiaffine_example_lower_order}
    {\mathbf{r}'''}^2 - 2 \mathbf{r}'' \cdot \mathbf{r}^{(4)} + 2 \mathbf{r}' \cdot \mathbf{r}^{(5)} = \mathrm{const}
\end{equation}
or, after differentiating both sides
\begin{equation}\label{equation_equiaffine_example}
    \mathbf{r}' \cdot \mathbf{r}^{(6)} = 0\, .
\end{equation}
For curves in plane equations \eqref{equation_equiaffine_example_lower_order} and \eqref{equation_equiaffine_example} respectively have the following form:
\begin{equation}\label{equation_equiaffine_example_in_plane_lower_order}
        x'''^2 + y'''^2 - 2 \left( x'' x^{(4)} + y'' y^{(4)}\right) + 2 \left(x' x^{(5)} + y' y^{(5)} \right) = \mbox{const}\, ,
\end{equation}
\begin{equation}\label{equation_equiaffine_example_in_plane}
        x' x^{(6)} + y' y^{(6)}= 0\, .
\end{equation}
Equation \eqref{equation_equiaffine_example} states that the scalar product between the first and the 6-th order derivatives of the position vector with respect to path's arc is zero. Studies \cite{Polyakov_etc_2001, Polyakov:2006, Polyakov_et_al_B.Cyb:2009} used equation \eqref{equation_equiaffine_example} for  equi-affine arc in plane. Since derivations there are applicable for an arbitrary feasible (strictly monotonous) parametrization $\sigma(t)$, that is when
\begin{equation}\label{equation_strictly_monotoneous}
    \frac{d\sigma(t)}{dt} \neq 0,  \; 0\leq t \leq T  \; ,
\end{equation}
 equation \eqref{equation_equiaffine_example} is suitable for an arbitrary feasible geometric parametrization of a curve and was also used for curves parameterized with the spatial equi-affine arc \cite{Polyakov:2006, Polyakov_et_al_B.Cyb:2009}.
\par Consider the situation of looking for curves along which maximally smooth trajectories would accumulate geometric measurement $\sigma$ with constant rate. So a curve is not given but has to be identified using equation \eqref{equation_equiaffine_example}. Arbitrary vector functions $\mathbf{r}(\cdot)$ that are solutions of \eqref{equation_equiaffine_example} do not necessarily represent curves' parameterization with desirable measurement $\sigma$ as demonstrated further in example \ref{example:parametrization}. Therefore, in the case of looking for curves parameterized with planar and spatial equi-affine arcs equation \eqref{equation_equiaffine_example} was endowed with additional condition:
\begin{eqnarray}
    \label{eq:condition_planar_ea_vel_introduction} \;\;\;\;\;\mathrm{in\; plane:}\;\;\;\;\;\;x' y'' - x'' y' = 1\, ,\;\;\;\; \\
    \nonumber \\
    \label{eq:condition_spatial_ea_vel_introduction} \mathrm{in\; space:}\;\; \left |
     \begin{array}{ccc}
        x' & x'' & x''' \\
        y' & y'' & y''' \\
        z' & z'' & z''' \\
     \end{array}
  \right | = 1
\end{eqnarray}
guaranteeing that the solution parametrizes a curve with the planar or spatial equi-affine arc respectively.
So necessary condition for the vector functions to describe planar paths whose maximally smooth trajectories satisfy the 2/3 power-law model was obtained in the form of the system of two differential equations: \eqref{equation_equiaffine_example} and \eqref{eq:condition_planar_ea_vel_introduction}.
\par Maximally smooth trajectories along parabolic segments have constant equi-affine velocity and zero jerk cost \cite{Polyakov_etc_2001, Polyakov:2006, Polyakov_et_al_B.Cyb:2009}. Note that parabolas constitute the only equi-affine solution of equation \eqref{equation_equiaffine_example} whose coordinates are described by polynomials of $\sigma$ \cite{Polyakov:2001, Polyakov:2006, Polyakov_et_al_B.Cyb:2009}. In general, only curves whose coordinates are 5-th order polynomials of $\sigma$ can have zero jerk cost.
\vspace{0.3cm}
\par Any curve can be geometrically parameterized in infinitely many different ways. Widely known parametrization is based on the Euclidian arc-length. In general, geometric parametrization of a curve can be implemented via a continuous mapping of a scalar parameter onto the curve.
\begin{example}\label{example:length_weighted_with_degree_Euclidian_curvature} A curve without inflection points can be parameterized with the integral of Euclidian speed weighted with Euclidian curvature raised to certain degree: $\tilde{\sigma}_{\beta}(t) = \displaystyle\int_{0}^{t}\dot{\sigma}_{eu}(\tau) \cdot {\left[c_{eu}(\sigma_{eu}(\tau))\right]}^\beta d\tau$. This geometric parametrization is legitimate for any $\beta$. It is measurement of the Euclidian arc when $\beta = 0$ and measurement of the equi-affine arc when $\beta = 1/3$.  \;\;\; $\Box$
\end{example}
\vspace{0.2cm}
\par Geometric measurement $\tilde{\sigma}_{\beta}$ from example \ref{example:length_weighted_with_degree_Euclidian_curvature} is invariant under Euclidian transformations for any $\beta$. However when $\beta \neq 0$ it is not equal to Euclidian length.
\vspace{0.3cm}

 Most of the   results presented in this work were initially published as arXiv manuscript \cite{Polyakov_equation_arXiv:2014} in September, 2014. Some explanations and formulae appearing in the first version of arXiv manuscript were improved, corrected or removed and novel material was added in arXiv versions 2 and 3. The current version of the manuscript is result of significant revision. It contains novel solutions, some formulations and explanations have been corrected and in many parts of the text improved. The manuscript now contains discussion about performance of the compromised motor control system in the framework of the  theory under consideration.

\section*{Methods and Results}\label{section:Methods_and_results}
\subsection*{Systems of differential equations, general case}
\par Consider $n$ times differentiable vector function in $L$-dimensional space $\mathbf{r}_L(\sigma) = (x_1(\sigma),\, x_2(\sigma),\, \ldots,\, x_L(\sigma))$ and $n$ times differentiable function $\sigma(t)$ defined in the interval $t \in [0,\; T]$, and strictly monotonous\footnote{So that the inverse function $t = \tau(\sigma)$ can be defined.}, $\sigma(0) = 0$. Function $\sigma(t)$ may represent geometric arc defined based on the vector function $\mathbf{r}_L(\sigma(t))$ and its derivatives with respect to time but is not limited to this. For example, when $\sigma(t)$ represents length of a curve drawn up to time $t$, $\sigma(t) = \int_0^t \sqrt{\|\dot{\mathbf{r}}(\tau)\|^2} d\tau$. The mean squared derivative cost functional \cite{Richardson_Flash_2003} associated with $n$-th order derivative of the vector function and its geometric parametrization $\sigma(t)$ is defined as follows:
\begin{eqnarray}
  \nonumber  J_{\sigma}(\mathbf{r}_L,\, n) & = & \displaystyle\int\limits_0^T \left\{
\left[\frac{d^n x_1(\sigma(t))}{dt^n}\right]^2 + \left[\frac{d^n x_2(\sigma(t))}{dt^n}\right]^2 + \ldots + \left[\frac{d^n x_L(\sigma(t))}{dt^n}\right]^2 \right\}\, dt \\
\label{general_cost_function1} & = & \displaystyle\int\limits_0^T \left\|\frac{d^n}{dt^n} \mathbf{r}_L(\sigma(t)) \right\|^2\, dt \, .
\end{eqnarray}
I call the order of derivative $n$  in the cost functional \eqref{general_cost_function1} degree of smoothness or order of smoothness.
\subsubsection*{Movement paths providing optimal trajectories with constant rate of accumulating geometric measurement}
\par Given geometric parametrization $\sigma$ of the vector function $\mathbf{r}$ (that describes a curve along which the trajectory is ``drawn'') in $L$-dimensional space, I denote with $\tilde{\sigma}^{*}_{\mathbf{r}_L, n}(t)$ temporal parametrization of ``drawing'' the path which provides minimal cost $J_{\sigma^{*}}(\mathbf{r}_L,\, n)$ (in other words maximal smoothness of degree $n$) under constraints at the boundary:
\begin{equation}\label{solution_path_minimization_problem2}
    \left\{
        \begin{array}{l}
            \dot{\tilde{\sigma}}^{*}_{\mathbf{r}_L, n}(t) =
                \displaystyle\arg\min_{\substack{\dot{\sigma}(t)}} J_{\sigma}(\mathbf{r}_L, \, n),\; t\in[0,\, T]\,, \\ \\
            \dot{\sigma}(0) = \dot{\sigma}(T) = \Sigma/T,\; \displaystyle\left.\frac{d^k\sigma}{dt^k}\right|_{t = 0} = \left.\frac{d^k\sigma}{dt^k}\right|_{t = T}  = 0,\; k = 2,\, \ldots,\, n-1\, ,
        \end{array}
    \right.
\end{equation}
dot denotes differentiation with respect to time $t$, $\Sigma = \int_{0}^{T} \dot{\sigma}(\tau) d\tau$. The constraints at the boundary mean that the optimal rule of ``drawing'' along given curve has to be picked out of functions whose first derivative at the boundaries is equal to $\sigma(T) / T$ and whose higher order derivatives up to order $n - 1$ are zero at the boundaries. The solution of the optimization problem \textit{without} constraints at the boundaries is also considered and is denoted as $\sigma^{*}_{\mathbf{r}_L, n}(t)$:
\begin{equation}\label{solution_path_minimization_problem1}
     \dot{\sigma}^{*}_{\mathbf{r}_L, n}(t) =
     \arg\min_{\substack{\dot{\sigma}(t)}} J_{\sigma}(\mathbf{r}_L, \, n),\; t\in[0,\, T]\, .
\end{equation}
Solutions of the optimization problem without boundary conditions are allowed to have arbitrary speed and its derivatives at the boundaries. For a given dimension of the space $L$ and the degree of smoothness $n$ I aim to find curves for which the solution of the optimization problem with constraints \eqref{solution_path_minimization_problem2} provides constant speed of ``drawing'' the curve:
\begin{equation}\label{eq_goal_set2}
    \mathcal{\tilde{A}}_{n,\, L} = \left\{\mathbf{r}_L: \dot{\tilde{\sigma}}^{*}_{\mathbf{r}_L, \, n}(t) = \mathrm{const} = \frac{\Sigma}{T},\; t \in [0,\, T]\right\}\,.
\end{equation}
The problem of identifying the curves under the same optimality criterion but without the boundary conditions from \eqref{solution_path_minimization_problem2} is also considered:
\begin{equation}\label{eq_goal_set1}
  \mathcal{A}_{n,\, L} = \left\{\mathbf{r}_L: \dot{\sigma}^{*}_{\mathbf{r}_L, \, n}(t) = \mathrm{const} = \frac{\Sigma}{T},\; t \in [0,\, T]\right\}\, .
\end{equation}
\par Solutions of the optimization problem \eqref{solution_path_minimization_problem1} for the curves from the set $\mathcal{A}_{n,\, L}$ (no boundary conditions are established) satisfy the boundary conditions established in \eqref{solution_path_minimization_problem2} anyhow and both optimization problems minimize the same cost functional. Therefore all solutions (curves) belonging to the set $\mathcal{A}_{n,\, L}$ belong to the set $\mathcal{\tilde{A}}_{n,\, L}$:
 \begin{equation}\label{eq:set.belongs.to.set}
    \mathcal{A}_{n,\, L} \subset \mathcal{\tilde{A}}_{n,\, L}\;.
 \end{equation}
So the necessary conditions derived for the curves from the class $\mathcal{\tilde{A}}_{n,\, L}$ are obeyed actually by the curves from both classes: $\mathcal{\tilde{A}}_{n,\, L}$ and $\mathcal{A}_{n,\, L}$.
\par Introduce a system of two differential equations:
\begin{equation}\label{eq:general_form_lower_order}
    \left\{
    \begin{array}{l}
        \left\|\displaystyle\frac{d^n\mathbf{r}}{d\sigma^n}\right\|^2 + 2 \displaystyle\sum_{i = 1}^{n-1} (-1)^i \left( \frac{d^{n-i}\mathbf{r}}{d\sigma^{n-i}} \cdot \frac{d^{n+i} \mathbf{r}}{d\sigma^{n+i}} \right) = \mathrm{const} \\
        \dot{\sigma}(t)|_{t(\sigma) = \sigma} \equiv v(t)|_{t(\sigma) = \sigma} = 1\; ,
    \end{array}
    \right.
\end{equation}
dot between two vectors denotes their scalar product.
Differentiation of both sides of the upper equation in \eqref{eq:general_form_lower_order} implies a system in which the upper equation is represented just by the scalar product of the first and $2n$-th order derivatives of the position vector with respect to the geometric parameter $\sigma$:
\begin{equation}\label{eq:general_form_higher_order}
    \left\{
    \begin{array}{l}
        \displaystyle\frac{d \mathbf{r}}{d\sigma} \cdot \frac{d^{2n} \mathbf{r}}{d\sigma^{2n}} = 0 \\ \\
        \dot{\sigma}(t)|_{t(\sigma) = \sigma} \equiv v(t)|_{t(\sigma) = \sigma} = 1 \; .
    \end{array}
    \right.
\end{equation}
% The lower row of the systems \eqref{eq:general_form_lower_order}, \eqref{eq:general_form_higher_order} represents the condition that the vector function satisfying the equation in the upper row represents a curve.
\par I consider the curves from the classes $\mathcal{\tilde{A}}_{n,\, L}$ and $\mathcal{A}_{n,\, L}$ defined in \eqref{eq_goal_set2} and \eqref{eq_goal_set1} respectively as candidates for geometric movement primitives. The systems \eqref{eq:general_form_lower_order}, \eqref{eq:general_form_higher_order} can be used as a tool for identifying such curves as follows from the main mathematical result of this work:
\begin{prop}\label{main_proposition} The curves along which optimal trajectories having degree of smoothness $n$ accumulate geometric measurement $\sigma(t)$ with constant rate, that is the curves from the sets  $\mathcal{\tilde{A}}_{n,\, L}$ and $\mathcal{A}_{n,\, L}$, satisfy the upper equations in the systems \eqref{eq:general_form_lower_order} and \eqref{eq:general_form_higher_order}.
\end{prop}
%
%%% The curves from Proposition \ref{main_proposition} belong to the sets $\mathcal{\tilde{A}}_{n,\, L}$ and $\mathcal{A}_{n,\, L}$ .
%
A detailed proof of proposition \ref{main_proposition} is provided in Appendices A, B.
Particular cases of the system \eqref{eq:general_form_higher_order} together with known solutions are demonstrated for different geometric parameterizations further in text %  in sections \ref{section:examples_equation_different_geometries}, \ref{section:3d_equi_affine}
in equations \eqref{eq:2d_case_equiaffine}, \eqref{eq:2d_case_equiaffine_via_eq_aff_curvature}, \eqref{eq:2d_case_affine}, \eqref{eq:2d_case_centeraffine}, \eqref{eq:2d_case_equicenteraffine}, \eqref{eq:2d_case_Euclidian}, \eqref{eq:2d_case_similarity}, \eqref{eq:3d_case_equiaffine_higher_order}.
\vspace{0.3cm}

\par Now a number of important notes about the systems \eqref{eq:general_form_lower_order} and \eqref{eq:general_form_higher_order}.
\begin{enumerate}
 \item \textbf{The upper equation} in the systems \eqref{eq:general_form_lower_order} and \eqref{eq:general_form_higher_order} is independent on the geometric parametrization of a curve. The equation is derived from two criteria: (1) maximal smoothness \eqref{solution_path_minimization_problem2} or \eqref{solution_path_minimization_problem1} during accumulating $\sigma(t)$ along the vector function $\mathbf{r}_L(\sigma)$ and (2) constancy of the rate of accumulating $\sigma(t)$ \eqref{eq_goal_set2} or \eqref{eq_goal_set1}. Derivation of the upper equation is based on the Euler-Poisson equation for variational problems. Particular cases of the upper differential equation for planar curves when $2 \leq n \leq 4$ and for an arbitrary $n$ are provided\footnote{Cost functionals $J_{\sigma}(\mathbf{r}_L,\, n)$ for the planar (L = 2) and spatial (L = 3) curves were used in different motor control studies for orders of differentiation $n$ equal to 2-4, eg. \cite{Hogan:1984, Flash_Hogan_1985, Viviani_Flash_1995, Todorov_Jordan_1998, Polyakov:2001, Polyakov_etc_2001, Richardson_Flash_2003, Polyakov:2006, Ben_Itzkah.Karniel:2008, Polyakov_et_al_B.Cyb:2009}.} in Table \ref{Table_equations}. Orthogonality, and consequently the left hand side of the upper equation in the systems, are invariant under arbitrary Euclidian transformations, uniform scalings and reflections.
     \par  Relationship $\bm{r}(\sigma)$ can be substituted directly into the upper equation of \eqref{eq:general_form_higher_order}. Nevertheless usually curves used in different studies are parameterized by polar angle, length, one of the coordinates (e.g. $y = y(x)$ for a plane curve), etc, and not necessarily by the desirable geometric measurement. However there is no need to write explicit relationship $\bm{r}(t(\sigma))$ between the coordinates of the curve $\bm{r}(t)$ and the geometric parameter $\sigma$ in order to compute derivatives in the upper equation of \eqref{eq:general_form_higher_order} and to remain consistent with prescribed parametrization $\sigma$. If $t$ is an arbitrary argument and $\dot{\sigma}$ is never zero then all derivatives of $\bm{r}(t(\sigma))$ with respect to $\sigma$ can be computed by applying recursively the chain rule and derivative of inverse function:
    \begin{equation}\label{eq:formulae_substituttions}
        \mathbf{r}'(\sigma)|_{\sigma = \sigma(t)} = \frac{\dot{\bm{r}}(t)}{\dot{\sigma}}\,,\;
        \mathbf{r}''(\sigma)|_{\sigma = \sigma(t)} = \frac{d\left( \frac{\dot{\bm{r}}(t)}{\dot{\sigma}} \right) / dt}{\dot{\sigma}} \,,\; \ldots \; .
    \end{equation}
     Assume a given curve is being tested for being a candidate primitive shape (belonging to the classes $\mathcal{\tilde{A}}_{n,\, L}$, $\mathcal{A}_{n,\, L}$) for concrete desirable $\dot{\sigma}$. Then derivatives in \eqref{eq:formulae_substituttions} can be computed in a program for symbolic computations, e.g. Mathematica, and the constraint formalized by the lower equation in the systems \eqref{eq:general_form_lower_order} and \eqref{eq:general_form_higher_order} can be easily plugged into the upper equation.
 \item In case one aims to find a novel candidate curve for geometric primitive \textbf{the lower equation} in the systems \eqref{eq:general_form_lower_order} and \eqref{eq:general_form_higher_order} is needed to guarantee that the solution of the parameter-independent upper equation indeed represents a curve with required geometric parametrization and provides consequent geometric invariance of the parametrization. In this case of unknown candidate curve use of substitution \eqref{eq:formulae_substituttions} would be generally impractical to my view and system of two equations has to be solved. The lower equation depends on the choice of geometric parametrization and essentially represents functional form of its derivative with respect to an arbitrary argument. Example \ref{example:parametrization} below demonstrates a curve that satisfies the upper equation for equi-affine parametrization and is not consistent with some other parameterizations.
 \item Use of \textbf{transversality conditions} in addition to the upper equation in the systems \eqref{eq:general_form_lower_order} and \eqref{eq:general_form_higher_order} can provide a more restricted necessary condition for a curve to belong to the set $\mathcal{A}_{n,\, L}$.
 \item The following \textbf{sufficient condition} for the curves from the set $\mathcal{\tilde{A}}_{3,\, L}$ was formulated earlier for the case of the minimum-jerk cost ($n = 3$) \cite{Polyakov:2006}:
     \begin{equation}\label{eq:sufficient_condition_n_is_3}
        \left\{
            \begin{array}{l}
                {\mathbf{r}'''}^2 - 2 \mathbf{r}'' \cdot \mathbf{r}^{(4)} + 2 \mathbf{r}' \cdot \mathbf{r}^{(5)} = \mathrm{const}_0 \geq 0 \\
                \displaystyle\min_{0 \leq \sigma \leq \Sigma}\left[ 9 {\mathbf{r}'''}^2 + 2 \mathbf{r}' \cdot \mathbf{r}''' - 24 \mathbf{r}'' \cdot \mathbf{r}''' \right] = \mathrm{const}_1 \geq 0 \; .
            \end{array}
        \right.
     \end{equation}
     The left hand side of the upper equations in the system \eqref{eq:sufficient_condition_n_is_3} is identical to the left hands side of upper equation in the necessary condition \eqref{eq:general_form_lower_order} and is restricted with the constraint on the sign of the constant in the right hand side.
 \end{enumerate}
\begin{landscape}
\begin{table}[h]
% \begin{adjustwidth}{-0.5in}{0.0in}
\vspace*{-2cm}

\hspace*{0cm}
\vspace*{1cm}
% \begin{sideways}
\begin{landscape}
\begin{tabular}{|c|c|c|l|l|}
\hline
 \textbf{Order} & \textbf{Equation, exemplar or} & \textbf{Derivative of}  & \textbf{Comments about} & \textbf{Known solutions when} $\sigma$ \\
       & \textbf{general case}          & \textbf{the equation}   & \textbf{  equation }  & \textbf{is equi-affine arc} \\
 \hline
 \textbf{2} & ${x''}^2 + {y''}^2 - 2x'x^{(3)} - 2y'y^{(3)}$  & $x' x^{(4)} + y' y^{(4)} = 0$  & Planar ``minimum- &  \\
 & $ = \mathrm{const}$ & & acceleration'' criterion  & Parabolas \\
 & & & in motor control & \\
 \hline
 \textbf{3} & ${x'''}^2 + {y'''}^2 - 2x''x^{(4)} - 2y''y^{(4)}$  & $x' x^{(6)} + y' y^{(6)} = 0$  & Planar ``minimum- & 2D: Parabolas, circles \cite{Polyakov:2006, Polyakov_et_al_B.Cyb:2009}, \\
 &&& jerk'' criterion & logarithmic spiral \cite{Bright:2006, Polyakov_et_al_B.Cyb:2009}, \\ &  $ + 2x'x^{(5)} + 2y'y^{(5)} = \mathrm{const}$  & &  in motor control & 3D: Parabolic screw line \cite{Polyakov:2006, Polyakov_et_al_B.Cyb:2009}, \\
    & & & & 3D: Elliptic screw line\\
 \hline
 \textbf{4} &  $\left(x^{(4)}\right)^2 + \left(y^{(4)}\right)^2$  & $x' x^{(8)} + y' y^{(8)} = 0$  &  Planar ``minimum- & 2D: Parabolas, circles,\\
 & $  - 2x'''x^{(5)} - 2y''y^{(5)} $ & & snap'' criterion & logarithmic spiral\\
   & $ + 2x''x^{(6)} + 2y''y^{(6)}$               &                                & in motor control & 3D: Parabolic \& elliptic screw \\
   & $ - 2 x' x^{(7)} - 2 y' y^{(7)} = \mathrm{const} $ & &  & lines \\
 \hline
\ldots   & \ldots & \ldots  &  & \\
   \hline
 \textbf{\it{n}} & $\left(x^{(n)}\right)^2 + \left(y^{(n)}\right)^2 + $ & $x' x^{(2n)} + y' y^{(2n)} = 0 $  &  Planar equation & \\
  & $2 \displaystyle\sum_{i = 1}^{n-1} (-1)^i (x^{(n-i)} x^{(n+i)}  $ & & &   \\
  & $+ y^{(n-i)} y^{(n + i)})  = \mathrm{const} $& & &  For $L <= n$, $\mathbf{r} =  \{x_1,\, \ldots,\, x_L\}$ \\
     &  $\|\mathbf{r}^{(n)}\|^2 $  & $\mathbf{r}' \cdot \mathbf{r}^{(2 n)} = 0$    &   & s.t. $x_k(\sigma) = \displaystyle\frac{\sigma^k}{k!}$, \\
& $ + 2 \displaystyle\sum_{i = 1}^{n-1} (-1)^i \left( \mathbf{r}^{(n-i)} \cdot \mathbf{r}^{(n+i)} \right)$ & &      $L$-dimensional space & $k = 1,\, \ldots,\, L$, $\sigma$ is \\
& $ = \mathrm{const} $ & & & equi-affine arc in dimension $L$ \\
 \hline
\end{tabular}
\end{landscape}
% \end{sideways}
% \end{adjustwidth}
\caption{\small{\bf Equations corresponding to different orders of smoothness in dimensions 2, 3, 4, n.} Upper equation in the systems \eqref{eq:general_form_lower_order}, \eqref{eq:general_form_higher_order} which is necessarily satisfied by the curves belonging to the classes $ \mathcal{\tilde{A}}_{n,\, L}$ and $\mathcal{A}_{n,\, L}$ defined in \eqref{eq_goal_set2} and \eqref{eq_goal_set1} respectively. Prime and order of differentiation in the brackets correspond to the derivative with respect to $\sigma$. Dot between two vectors in the row corresponding to the case $n$ denotes scalar product of the vectors. Details about planar solutions mentioned in Table \ref{Table_equations} are provided in Table \ref{Table_mathematical_objects}. The spatial solutions are analyzed in the part related to the spatial equi-affine group.}\label{Table_equations}
\end{table}
\end{landscape}

\par In a different study Meirovitch has recently introduced a variational problem for plane jerk without explicitly incorporating rate of accumulating arc ($\dot{\sigma}$) and with constraint in terms of path instead of relationship between movement speed and duration; such representation enabled him to find an elegant proof of the necessary condition on a given planar trajectory that minimizes jerk cost and follows prescribed path:
\begin{equation}\label{eq:Meirovitch}
    \dot{x} \frac{d^{6}x}{dt^6} + \dot{y} \frac{d^{6}y}{dt^6} = 0\, ,
\end{equation}
in other words being a necessary condition for a given trajectory satisfying the constrained minimum-jerk model\footnote{Personal communication (August, 2015), apparently this result appears in Supplementary Material to \cite{Meirovitch:2014}.}. %Actually this elegant proof can be easily extended to the case of spatial trajectories by representing curve as intersection of two surfaces and introducing additional Lagrange multiplier.
Given that $\sigma(t) = t$ is always a legitimate parametrization for a path of a given trajectory as condition \eqref{equation_strictly_monotoneous} is satisfied for such $\sigma(t)$, equation \eqref{eq:Meirovitch} is equivalent to equation \eqref{equation_equiaffine_example_in_plane}, that is the plane version of the upper equation in \eqref{eq:general_form_higher_order} with $n = 3$, for every given trajectory (that itself specifies movement path). The proof by Meirovitch, though elegant, does not lead directly to the first integral \eqref{equation_equiaffine_example_in_plane_lower_order} of equation \eqref{equation_equiaffine_example_in_plane}. Equation\eqref{equation_equiaffine_example_in_plane} has recently been used for mixture of geometries approach \cite{Meirovitch:2014}.
\subsubsection*{ The process of identifying candidates for geometric primitives based on proposed system \eqref{eq:general_form_lower_order} may follow either of the three approaches}
\begin{enumerate}
    \item\label{lbl:approaches_candidate_solutions1} Specify desirable (geometric) parametrization $\sigma$ and solve the system \eqref{eq:general_form_lower_order} (or \eqref{eq:general_form_higher_order}) to identify candidate curves along which optimal trajectories conserve time derivative of prescribed geometric parameter, for example equi-affine or affine arcs. Here the lower equation of the systems is used as is.
    \item\label{lbl:approaches_candidate_solutions2} Specify desirable (geometric) parametrization $\tilde{\sigma}$ and guess candidate curves among solutions of the upper equations of either of the systems \eqref{eq:general_form_lower_order}, \eqref{eq:general_form_higher_order}:
        \begin{equation}\label{eq:general_form_lower_order_versus_t}
            \left\|\displaystyle\frac{d^n\mathbf{r}}{d\sigma^n}\right\|^2 + 2 \displaystyle\sum_{i = 1}^{n-1} (-1)^i \left( \frac{d^{n-i}\mathbf{r}}{d\sigma^{n-i}} \cdot \frac{d^{n+i} \mathbf{r}}{d\sigma^{n+i}} \right) = \mathrm{const}\, ,
        \end{equation}
    \begin{equation}\label{eq:only_upper_equation_higher_order}
        \displaystyle\frac{d \mathbf{r}}{d\sigma} \cdot \frac{d^{2n} \mathbf{r}}{d\sigma^{2n}} = 0\, ,
    \end{equation}
    where $\sigma$ is an arbitrary parametrization including the possibility $\sigma = t$.
    For example such solutions as polynomials of degree $\leq 2n - 1$ with respect to $\sigma$ or certain periodic functions can be easily guessed. Then among the guessed candidate curves determine those (if any) consistent with both equation \eqref{eq:general_form_lower_order_versus_t} and desired parametrization $\tilde{\sigma}$, see also example \ref{example:parametrization}. Segregation of the guessed curves is straightforward by either (a) parameterizing them with desired $\tilde{\sigma}$ or (b) by using substitution \eqref{eq:formulae_substituttions}; for a known curve both (a) and (b) may be viewed as replacement of the condition formalized by the lower equation in \eqref{eq:general_form_lower_order}.
\item\label{lbl:approaches_candidate_solutions3} Do not request any specific parametrization of candidate curves. Search for candidate curves that are described by arbitrary vector functions satisfying equation \eqref{eq:general_form_lower_order_versus_t} (or \eqref{eq:only_upper_equation_higher_order}), remain with parameterizations induced by identified solutions and possibly look for their geometric meanings. For example whether $\sigma$ is proportional to some geometric arc or another geometric parameter computed for a curve. Here the constraint formalized by the lower equation of the systems \eqref{eq:general_form_lower_order}, \eqref{eq:general_form_higher_order} does not exist from the beginning but its equivalent\footnote{Meaning formula for $\dot{\sigma}$ induced by its geometric properties and based on differentiation with respect to an arbitrary parameter. Different examples of the lower equation can be found further in text.} might be derived.
\end{enumerate}
In earlier studies \cite{Polyakov:2001, Polyakov:2006, Bright:2006, Polyakov_et_al_B.Cyb:2009} and earlier versions of this manuscript (1 - 3) known candidate curves were found using educated guess in Approaches \ref{lbl:approaches_candidate_solutions1} and \ref{lbl:approaches_candidate_solutions2}. In the present version of the manuscript method \ref{lbl:approaches_candidate_solutions3} is also used to identify additional candidate curves \eqref{eq:solutions_polynomials}, \eqref{eq:solutions_spirals_polynomials} without prior knowledge about geometric meaning of the parameter making these curves solutions of the system \eqref{eq:general_form_lower_order}.
\par Future studies may create a machinery for solving systems \eqref{eq:general_form_lower_order}, \eqref{eq:general_form_higher_order} and make Approach \ref{lbl:approaches_candidate_solutions1} a handy tool for identifying candidate curves parameterized by a measurement with desired properties. Equation \eqref{eq:2d_case_equiaffine_via_eq_aff_curvature} further in text demonstrates example of what could become a part of Approach \ref{lbl:approaches_candidate_solutions1}: the planar system of equations \eqref{eq:2d_case_equiaffine} for $n = 3$ is written in the form of a single equation in terms of curve's equi-affine curvature. Equation \eqref{eq:2d_case_equiaffine_via_eq_aff_curvature} can be applied whenever one wishes to characterize candidate curves with their equi-affine curvature.

\begin{example}\label{example:parametrization}
%\footnote{Necessary definitions from affine differential geometry are provided in the text. More detailed background for the geometric notions which are used in example \ref{example:parametrization} can be found elsewhere, eg. in \cite{Shirokovy_1959, Guggenheimer:1977}. Translation of parts of \cite{Shirokovy_1959} and notions of equi-affine geometry  within the framework of motor control can be found in \cite{Polyakov:2006}. Relationship between equi-affine geometry and motor control was established in \cite{Pollick_Shapiro_1997, Flash_Handzel:2007}.}
    Time derivative of the equi-affine arc called equi-affine velocity is computed as follows: $\dot{\sigma}_{ea} \equiv v_{ea} = (\dot{x} \ddot{y} - \dot{y} \ddot{x}) ^ {1/3}$. Time derivative of the Euclidian arc is computed as $\dot{\sigma}_{eu} \equiv v_{eu} = \sqrt{\dot{x}^2 + \dot{y}^2}$. Computation of Euclidian curvature at some point of the trajectory is as follows:
    $$c_{eu} = \frac{\dot{x} \ddot{y} - \dot{y} \ddot{x}}{(\dot{x}^2 + \dot{y}^2)^{3/2}} = \frac{ v_{ea}^3}{v_{eu}^3}$$
    and therefore
    $$v_{ea} = v_{eu} \, c_{eu}^{1/3} \, .$$
    Consider parametrization with cumulative Euclidian speed weighted with Euclidian curvature raised to the power $\beta$:
    \begin{equation}\label{speed_parametrized_beta}
        \tilde{v}_{\beta} \equiv v_{eu} \, {c_{eu}}^{\beta} = (\dot{x}^2 + \dot{y}^2)^{1/2 -  3 \beta / 2}\, (\dot{x} \ddot{y} - \dot{y} \ddot{x}) ^ {\beta} \; .
    \end{equation}
    Corresponding geometric measurement is equal to the integrated speed:
    $$\tilde{\sigma}_{\beta}(t) = \displaystyle\int_{\tau = 0}^{t} \tilde{v}_{\beta} d\tau\,  $$
    which is strictly monotonous function for the curve without inflection points (though it does not necessarily represent an arc of a curve in some geometry). Therefore $\tilde{\sigma}_{\beta}$ is legitimate geometric parametrization of a curve without inflection points. Obviously, the vector function
    \begin{equation}\label{parabola_canonical_form_example}
        x = \tilde{\sigma}_{\beta},\; y = \tilde{\sigma}_{\beta}^2 / 2
    \end{equation}
    satisfies equations in Table \ref{Table_equations}  for $n \geq 2$ as (2n)-th order derivatives of $x$, $y$ with respect to $\tilde{\sigma}_{\beta}$ are all zero. Such vector function describes parabola parameterized with equi-affine arc, case of $\beta = 1/3$. If $\tilde{\sigma}_{\beta}(t)$ linearly depends on $t$ meaning constant speed $\alpha$ then $x = \alpha t,\; y = \alpha^2 t^2/2$ implying $\dot{x} \ddot{y} - \dot{y} \ddot{x} = \alpha^3$. In turn, substituting $t = \tilde{\sigma}_{\beta} / \alpha$ into \eqref{speed_parametrized_beta} implies
    $$\dot{\tilde{\sigma}}_{\beta} \equiv \tilde{v}_{\beta} = (\dot{x}^2 + \dot{y}^2)^{1/2 -  3 \beta / 2} \cdot \alpha^{3 \beta} = (1 + {\tilde{\sigma}_{\beta}}^2)^{1/2 -  3 \beta / 2} \cdot \alpha^{3 \beta} \neq \alpha \; \mathrm{when}\,  \beta \neq 1/3\, .$$
    This long way to show that parabola cannot be parameterized with $\tilde{\sigma}_{\beta}$ in canonical form \eqref{parabola_canonical_form_example} whenever $\beta \neq 1/3$  was chosen for pedagogical reason. In particular, if one uses Approach \ref{lbl:approaches_candidate_solutions2} to select curves along which maximally smooth trajectories accumulate Euclidian arc with constant rate, parabolas would be filtered out with the lower equation of system \eqref{eq:general_form_higher_order}.  \;\;\; $\Box$
\end{example}

\subsubsection*{The problem of deriving optimal temporal profile for a given path}
\par Let some path $\bm{r}(s)$ be given. Consider the optimization problem \eqref{solution_path_minimization_problem1} of finding optimal trajectory $\bm{r}(s(t))$ along the known path. Similar to the constrained minimum-jerk model, optimal $s(t)$ should be found, but now with degree of smoothness $n$.
Equation \eqref{eq:only_upper_equation_higher_order} is satisfied for such problem with arbitrary feasible $\sigma$. In the spirit of the form \eqref{eq:Meirovitch} of planar equation \eqref{equation_equiaffine_example_in_plane} an alternative notation can be adopted in equation \eqref{eq:only_upper_equation_higher_order} for convenience of analyzing this specific problem:
% eq:only_upper_equation_higher_order
\begin{equation}\label{eq:only_upper_equation_higher_order_s_vs_t}
    \displaystyle\frac{d \mathbf{r}(s(t))}{d t} \cdot \frac{d^{2n} \mathbf{r}(s(t))}{d t^{2n}} = 0
\end{equation}
or simply
\begin{equation}\label{eq:only_upper_equation_higher_order_vs_t}
    \displaystyle\frac{d \mathbf{r}(t)}{d t} \cdot \frac{d^{2n} \mathbf{r}(t)}{d t^{2n}} = 0\, .
\end{equation}
Form \eqref{eq:only_upper_equation_higher_order_vs_t} of \eqref{eq:only_upper_equation_higher_order} is true for any trajectory $\bm{r}(t)$ constrained by some path and minimizing cost \eqref{general_cost_function1} because feasibility condition \eqref{equation_strictly_monotoneous} is satisfied for $\sigma(t) = t$.
\par Now the rule for derivative of the nested function can be applied to equation \eqref{eq:only_upper_equation_higher_order_s_vs_t} in order to formulate the optimization problem \eqref{solution_path_minimization_problem1} as differential equation for $s(t)$, generally non-linear.
\begin{example}\label{example:find_trajectory_given_path}
    \par Consider an ellipse $x(\theta) = a \cos\theta$, $y(\theta) = b \sin\theta$ and the problem of finding an optimal rule of drawing an ellipse with $n = 2$, that is constrained minimum-acceleration problem. Now use recursively $df(\sigma(t))/dt = f' \dot{\sigma}$ to derive: $d^4 f(\sigma(t))/dt^4 = f^{(4)} \dot{\sigma}^4 + 6 f''' \ddot{\sigma} \dot{\sigma}^2 + 4 x'' \dddot{\sigma} \dot{\sigma} + 3 x'' \ddot{s}^2 + x' d^4s/dt^4$. Equation \eqref{eq:only_upper_equation_higher_order_s_vs_t} with $n = 2$ becomes for the ellipse:
    $$\begin{array}{l}
         -a^2 \dot{\theta} \sin\theta \cdot \left(  \dot{\theta}^4 \cos\theta + 6 \ddot{\theta} \dot{\theta}^2 \sin\theta  - 4 \dddot{\theta} \dot{\theta} \cos\theta  - 3 \ddot{\theta}^2 \cos\theta  - d^4\theta/dt^4  \sin\theta \right)  +  \\
     b^2 \dot{\theta} \cos\theta \cdot \left( \dot{\theta}^4 \sin\theta - 6 \ddot{\theta} \dot{\theta}^2 \cos\theta  - 4 \dddot{\theta} \dot{\theta} \sin\theta  - 3 {\ddot{\theta}}^2 \sin\theta  + d^4\theta/dt^4  \cos\theta \right)  = 0 \, .
    \end{array}$$
 \;\;\; $\Box$
\end{example}
Example \ref{example:find_trajectory_given_path} demonstrates that generally it may be impractical to apply equation \eqref{eq:only_upper_equation_higher_order_s_vs_t} to derive an optimal speed of ``drawing'' along given path.
\subsection*{Different parameterizations, arcs in the geometries of affine group in plane and some of its subgroups, equations and solutions}\label{section:examples_equation_different_geometries}
\par Different kinds of invariance were analyzed in the studies of action and perception of motion. For example, point-to-point hand movements are assumed to produce nearly straight paths. Straight trajectories parameterized with Euclidian arc \eqref{Euclidian_speed} and having constant Euclidian velocity $\dot{\sigma}_{eu}$ satisfy systems \eqref{eq:general_form_lower_order} and \eqref{eq:general_form_higher_order} with arbitrary degree of smoothness. Straight paths are meaningful only in Euclidian geometry among the six geometries considered in this study\footnote{Equi-affine and similarity arcs of straight segments are zero, center-affine and affine arcs are not defined, equi-center-affine arc is zero when straight line crosses the origin (the arcs are introduced below).}. More complex movements were analyzed in the frameworks of equi-affine and affine geometries \cite{Pollick_Shapiro_1997, Handzel_Flash_1999, Polyakov:2001, Polyakov_etc_2001, Polyakov:2006, Flash_Handzel:2007, Polyakov_et_al_B.Cyb:2009, Polyakov_et_al_PLoS_C_B:2009, Bennequin:2009}. In this section I present system  \eqref{eq:general_form_higher_order} for different degrees of smoothness $n$ while assuming constant speed of accumulating arc in different geometries. The expressions for the rate of accumulating arc that are presented below are based on the results from very useful book by Shirokov \& Shirokov \cite{Shirokovy_1959}. Information about the relationship between the derived system of equations \eqref{eq:general_form_higher_order} and candidate solutions is summarized in Table \ref{Table_mathematical_objects}.

\par Explicit relationship $x(\sigma),\; y(\sigma),\; z(\sigma)$ is shown below for some of the curves in a number of geometries, $z$ is not relevant for plane curves of course. The expressions for $\dot{\sigma}(t)$ are provided for every geometry/parametrization under consideration.

\subsubsection*{Equi-affine group}\label{subsection:equi-affine_group}
\par Equi-affine transformations of coordinates involve 5 independent parameters and are of the form:
\begin{equation}\label{eq:equi-affine_transformation}
  \begin{array}{c}
      x_1 = \alpha x + \beta y + a \\
      y_1 = \gamma x + \delta y\, + b , \
  \end{array}
  \quad \left |
     \begin{array}{cc}
        \alpha & \beta \\
        \gamma & \delta
     \end{array}
  \right |  = 1\, .
\end{equation}
The rate of accumulating equi-affine arc is called equi-affine velocity. It is computed as follows \cite{Shirokovy_1959}:
\begin{equation}\label{eq:ea_speed}
    \dot{\sigma}_{ea} = \left |
     \begin{array}{cc}
        \dot{x} & \ddot{x} \\
        \dot{y} & \ddot{y}
     \end{array}
  \right | ^{1 / 3}\; .
\end{equation}
System \eqref{eq:general_form_higher_order} in plane becomes:
\begin{equation}\label{eq:2d_case_equiaffine}
    \left\{
    \begin{array}{l}
         x' x^{(2n)} +  y' y^{(2n)}  = 0 \\
        x' y'' - x'' y' = 1 \, .
    \end{array}
    \right.
\end{equation}
The lower equation guarantees that the solution of the upper equation is parameterized with the equi-affine arc
\begin{equation}\label{eq:ea_arc}
    \sigma_{ea} = \int_0^t \dot{\sigma}_{ea}(\tau) d\tau\; .
\end{equation}
Earlier studies used system \eqref{eq:2d_case_equiaffine} for the minimum-jerk cost functional ($n = 3$) \cite{Polyakov:2001, Polyakov:2006, Polyakov_et_al_B.Cyb:2009}. Here concrete curves are considered as candidate solutions, some of them are filtered out for the prescribed equi-affine parametrization (Approach \ref{lbl:approaches_candidate_solutions2}). The lower equation is automatically satisfied when a curve is parameterized with desired geometric measurement or when substitution \eqref{eq:formulae_substituttions} is used for such measurement.
\par The 1st and the 3rd order derivatives of the position vector of a planar curve $\mathbf{r}(\sigma_{ea})$ with respect to the equi-affine arc are parallel. The parallelism follows from the identity $x' y'' - x'' y' = 1$ which, in particular, appears in the system \eqref{eq:2d_case_equiaffine}. Apparently, the equi-affine curvature of a curve \cite{Shirokovy_1959, Guggenheimer:1977, Calabi_Olver_Tannenbaum_1996}%\footnote{More background for the current framework is provided in \cite{Polyakov:2006, Polyakov_et_al_B.Cyb:2009}.}
  \begin{equation}\label{eq:ea_curvature}
        \kappa_{ea}(\sigma_{ea}) = x''(\sigma_{ea}) y'''(\sigma_{ea}) - x'''(\sigma_{ea}) y''(\sigma_{ea}) \;
  \end{equation}
is a scaling factor between the 1st and the 3rd order derivatives of the position vector: $\mathbf{r}'''(\sigma_{ea}) + \kappa(\sigma_{ea}) \mathbf{r}'(\sigma_{ea}) = 0$
which implies the possibility to express higher order derivatives of the vector $\mathbf{r}(\sigma_{ea})$ in terms of its 1st and 2nd order derivatives when equi-affine curvature is a known function of the equi-affine arc. In particular, system \eqref{eq:2d_case_equiaffine} with $n = 3$ can be rewritten. Given that
\begin{equation}\label{eq:ea_curvature_scaling_factor}
    \mathbf{r}'''(\sigma_{ea}) = -\kappa_{ea}(\sigma_{ea}) \mathbf{r}'(\sigma_{ea})\, ,
\end{equation}
\begin{equation*}
    \mathbf{r}^{(6)}(\sigma_{ea}) = \mathbf{r}''(\sigma_{ea}) (\kappa_{ea}^2(\sigma_{ea}) - 3\kappa_{ea}''(\sigma_{ea})) + \mathbf{r}'(\sigma_{ea})(4 \kappa_{ea}'(\sigma_{ea}) \kappa_{ea}(\sigma_{ea}) - \kappa_{ea}'''(\sigma_{ea}))\;
\end{equation*}
and the upper equation of the system \eqref{eq:2d_case_equiaffine} for the case of $n = 3$ ($\mathbf{r}' \cdot \mathbf{r}^{(6)} = 0$) becomes
\begin{equation}\label{eq:2d_case_rprime_r6primes_via_eq_aff_curvature}
    \begin{array}{rcl}
        \mathbf{r}'(\sigma_{ea}) \cdot \mathbf{r}^{(6)}(\sigma_{ea}) & = & (\mathbf{r}'(\sigma_{ea}) \cdot \mathbf{r}''(\sigma_{ea})) (\kappa_{ea}^2(\sigma_{ea}) - 3\kappa_{ea}''(\sigma_{ea}))\\
            & + & {\mathbf{r}'}^2(\sigma_{ea})(4 \kappa_{ea}'(\sigma_{ea}) \kappa_{ea}(\sigma_{ea}) - \kappa_{ea}'''(\sigma_{ea})) = 0\; .
    \end{array}
\end{equation}
Noting that $\mathbf{r}'(\sigma_{ea}) \cdot \mathbf{r}''(\sigma_{ea}) = \displaystyle \frac{1}{2} \frac{d}{d\sigma_{ea}} ({\mathbf{r}'}^2(\sigma_{ea}))$, equation \eqref{eq:2d_case_rprime_r6primes_via_eq_aff_curvature} implies that
\begin{equation}\label{eq:equi_affine_der_squared_der}
    \frac{1}{2} \frac{d}{d\sigma_{ea}} ({\mathbf{r}'}^2(\sigma_{ea})) = {\mathbf{r}'}^2(\sigma_{ea}) \frac{4 \kappa_{ea}'(\sigma_{ea}) \kappa_{ea}(\sigma_{ea}) - \kappa_{ea}'''(\sigma_{ea})}{3\kappa_{ea}''(\sigma_{ea}) - \kappa_{ea}^2(\sigma_{ea})}\; .
\end{equation}
After integrating \eqref{eq:equi_affine_der_squared_der} the system \eqref{eq:2d_case_equiaffine} with $n = 3$ can be rewritten as follows:
\begin{equation}\label{eq:2d_case_equiaffine_via_eq_aff_curvature}
        {\mathbf{r}'}^2(\sigma_{ea}) = {\mathbf{r}'}^2(0) \exp\left[2 \left(F(\sigma_{ea}) -  F(0)\right)\right]\, ,
\end{equation}
where
\begin{equation*}
    F(\sigma_{ea}) = \int \frac{4 \kappa_{ea}' \kappa_{ea} - \kappa_{ea}'''}{3\kappa_{ea}'' - \kappa_{ea}^2} d\sigma_{ea} \; .
\end{equation*}
\par The expression for the equi-affine curvature can also be written for an arbitrary artument $t$ \cite{Shirokovy_1959}:
\begin{equation}\label{eq:ea_curvature_vs_t}
    \kappa_{ea}(t) = \sqrt{\frac{3 \dot{\sigma}_{ea}^3 \cdot \left |
     \begin{array}{cc}
        \dot{x} & d^4x / d t^4 \\
        \dot{y} & d^4y / d t^4
     \end{array}
  \right |
   + 12 \dot{\sigma}_{ea}^3 \cdot
    \left |
        \begin{array}{cc}
            \ddot{x} & \dddot{x} \\
            \ddot{y} & \dddot{y}
        \end{array}
    \right | -
    5 \left |
        \begin{array}{cc}
            \dot{x} & \dddot{x} \\
            \dot{y} & \dddot{y}
        \end{array}
    \right |^2
  } {9 \dot{\sigma}_{ea}^{8}}} \, .
\end{equation}
\vspace{0.5cm}

\par The results for candidate solutions (parabola, circle, logarithmic spiral) are as follows:
\begin{enumerate}
  \item \textbf{Parabola} is parameterized with equi-affine arc, up to an equi-affine transformation, as follows:
   \begin{equation}\label{eq:parabola_canonical}
        \begin{array}{l}
            x = \sigma_{ea} \\
            y = {\sigma_{ea}}^2 / 2 \; .
        \end{array}
  \end{equation}
  So $x$ and $y$ coordinates of the parabola in equi-affine parametrization are always polynomials of up to 2-nd degree with respect to $\sigma_{ea}$. The class of parabolas constitutes obvious solution of \eqref{eq:2d_case_equiaffine} for $n \geq 2$. The class of parabolas is invariant under arbitrary affine transformations \cite{Polyakov:2006, Polyakov_et_al_B.Cyb:2009}. Drawing parabolas with constant equi-affine velocity does minimize the cost functional \eqref{general_cost_function1} for $n \geq 2$ and provides zero cost. Corresponding results for $n = 3$ were reported in \cite{Polyakov:2001, Polyakov:2006, Polyakov_et_al_B.Cyb:2009}.
  \par Measuring path along parabola with equi-affine arc has the following interesting property.
   If, for three points $F$, $G$, $H$ on a
                    parabola, the chord $FH$ is parallel
                    to the tangent to the parabola at the intermediate point
                    $G$, then the equi-affine arc measured along parabola between
                    $F$ to $G$ equals to the equi-affine arc measured
                    between $G$ to $H$\footnote{The proof is provided in Appendix D of \cite{Polyakov:2006}.}. Such chord and tangent are demonstrated in Figure \ref{fig:3_points_parabola_chord_tangent}.
   \begin{figure}% [h!]
 \vspace*{-1.5cm}
% \hspace*{-0.5cm}
\begin{tabular}{c}
\psfig{figure=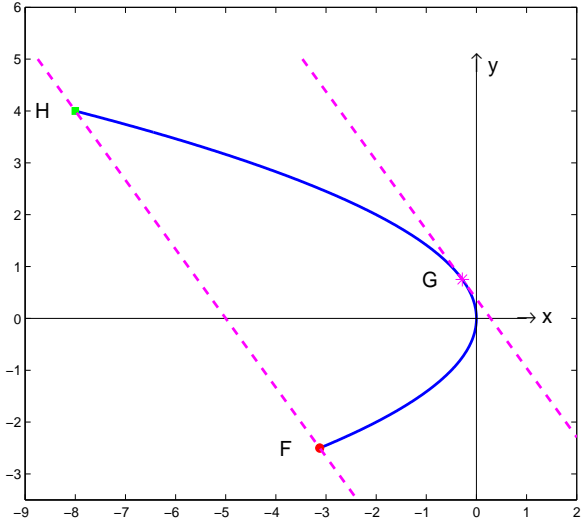,
height=70mm,width=77mm} %
\end{tabular}
\caption{\small \textbf{Tangent to a parabola and parallel chord are connected with path's portions having equal equi-affine arc.} Tangent to depicted parabolic segment at point $G$ is parallel to the chord $FH$. Equi-affine arc between $F$ and $G$ is equal to the equi-affine arc between $G$ and $H$. In the canonical coordinate system parabola is described by the relationship $x = -y^2/(2p)$, where $p$ is the focal parameter. If tangent is parallel to the chord then, in canonical coordinate system, $y$ coordinate of the point of touch by the tangent ($G$) is average of the $y$-coordinates of the chord ($FH$)}\label{fig:3_points_parabola_chord_tangent}
\end{figure}   
  \item \textbf{Circle} is parameterized with equi-affine arc as follows:
  \begin{equation}\label{eq:circle_equi_affine}
        \begin{array}{l}
        x = x_0 + \kappa_{ea}^{-3/4} \cdot \cos(\sqrt{\kappa_{ea}} \sigma_{ea}) \\
        y = y_0 + \kappa_{ea}^{-3/4} \cdot \sin(\sqrt{\kappa_{ea}} \sigma_{ea})\;.
        \end{array}
  \end{equation}
 For an arbitrary $n \geq 1$ circle is non-invariant solution under arbitrary equi-affine transformations, however circles are invariant solutions (remain circles) under Euclidian transformations  \eqref{eq:Eucl_transformation} composed with uniform scaling \eqref{eq:scaling_transformation_no_reflection} including reflection \eqref{eq:scaling_transformation_with_reflection}.
  \item \textbf{Logarithmic spiral} can be parameterized with polar angle:
     \begin{equation}\label{eq:log_spiral_polar}
        \begin{array}{l}
        x = x_{0} + \mathrm{const} \cdot e^{\beta \varphi} \cdot \cos{\varphi}\\
        y = y_{0} + \mathrm{const} \cdot e^{\beta \varphi} \cdot \sin{\varphi}\; .
        \end{array}
     \end{equation}
     Introducing parametrization with equi-affine arc from \eqref{eq:ea_speed} and integrating
      $$d\sigma_{ea} / d\varphi = (\mathrm{const} ^{2} (1 + \beta^2))^{1/3} \cdot  e^{2 \beta \varphi / 3}$$
     with initial condition $\varphi(0) = 0$ results in:
      $$\varphi(\sigma_{ea}) = \frac{3}{2 \beta} \ln\left( \frac{2\beta\sigma_{ea}}{3} \cdot (\mathrm{const} ^{2} (1 + \beta^2))^{-1/3} + 1 \right)\; $$
      which can be substituted into \eqref{eq:log_spiral_polar} to imply the expressions for $x(\sigma_{ea}),\, y(\sigma_{ea})$. Apparently, \eqref{eq:log_spiral_polar} is solution for \eqref{eq:2d_case_equiaffine} only for certain values of $\beta$ which depend on the degree of smoothness $n$. The values of $\beta$ corresponding to $n \leq 5$ are as follows.
      \begin{enumerate}
        \item $n = 1,\, 2$: no solutions.
        \item $n = 3$: $\beta = \pm 3 / \sqrt{7}$ (see also \cite{Bright:2006, Polyakov_et_al_B.Cyb:2009}) or $\beta \approx \pm 1.13389$.
        \item $n = 4$: $\beta = \pm \sqrt{\frac{43 \pm 4 \sqrt{97}}{33}}$ or $\beta \approx \pm 0.330499$, $\beta \approx \pm1.58014$.
        \item $n = 5$: $\beta = \pm \sqrt{\frac{3}{13}}$ and $\beta = \pm\sqrt{\frac{489 \pm 12 \sqrt{1609}}{275}}$ or $\beta \approx \pm 0.166808$, $\beta \approx \pm 0.480384$, $\beta \approx \pm 1.87844$.
      \end{enumerate}
\end{enumerate}
\vspace{0.3cm}

\subsubsection*{Affine group}
\par Planar affine transformations of coordinates involve 6 independent parameters and are of the form:
\begin{equation}\label{eq:affine_transformation}
  \begin{array}{c}
      x_1 = \alpha x + \beta y + a \\
      y_1 = \gamma x + \delta y\, + b , \
  \end{array}
  \quad \left |
     \begin{array}{cc}
        \alpha & \beta \\
        \gamma & \delta
     \end{array}
  \right | \neq 0\, .
\end{equation}
The speed of accumulating affine arc is computed as follows\footnote{The formula for the speed of accumulating affine arc parameterized with an arbitrary parameter $t$ (the middle formula in \eqref{eq:aff.speed}) in \cite{Shirokovy_1959} contains misprint and therefore it is different from \eqref{eq:aff.speed}. The rightmost formula in \eqref{eq:aff.speed} can be found in  \cite{Shirokovy_1959}.}:
\begin{equation}\label{eq:aff.speed}
    \dot{\sigma}_{a} = \sqrt{\frac{3 \dot{\sigma}_{ea}^3 \cdot \left |
     \begin{array}{cc}
        \dot{x} & d^4x / d t^4 \\
        \dot{y} & d^4y / d t^4
     \end{array}
  \right |
   + 12 \dot{\sigma}_{ea}^3 \cdot
    \left |
        \begin{array}{cc}
            \ddot{x} & \dddot{x} \\
            \ddot{y} & \dddot{y}
        \end{array}
    \right | -
    5 \left |
        \begin{array}{cc}
            \dot{x} & \dddot{x} \\
            \dot{y} & \dddot{y}
        \end{array}
    \right |^2
  } {9 \dot{\sigma}_{ea}^{6}}} = \dot{\sigma}_{ea} \sqrt{\kappa_{ea}}\; ,
\end{equation}
where $\dot{\sigma}_{ea}$ is equi-affine velocity \eqref{eq:ea_speed} and $\kappa_{ea}$ is equi-affine curvature \eqref{eq:ea_curvature}.
System \eqref{eq:general_form_higher_order} in plane becomes
\begin{equation}\label{eq:2d_case_affine}
    \left\{
    \begin{array}{l}
        x' x^{(2n)} +  y' y^{(2n)}  = 0 \\
        \displaystyle\frac{3 (x' y'' - x'' y') \cdot \left |
     \begin{array}{cc}
        x' & x^{(4)} \\
        y' & y^{(4)}
     \end{array}
  \right |
   + 12 (x' y'' - x'' y') \cdot
    \left |
        \begin{array}{cc}
            x'' & x''' \\
            y'' & y'''
        \end{array}
    \right | -
    5 \left |
        \begin{array}{cc}
            x' & x''' \\
            y' & y'''
        \end{array}
    \right |^2
  } {9 (x' y'' - x'' y')^{2}} =  1 \; .   \end{array}
    \right.
\end{equation}
Here concrete curves are considered as candidate solutions, some of them are filtered out for the prescribed affine parametrization (Approach \ref{lbl:approaches_candidate_solutions2}). The lower equation is automatically satisfied when a curve is parameterized with desired geometric measurement or when substitution \eqref{eq:formulae_substituttions} is used for such measurement.
The results for candidate solutions (parabola, circle, logarithmic spiral) are as follows:
\begin{enumerate}
  \item \textbf{Parabolas}' affine arc is zero, same as equ-affine arc of the straight line is zero or Euclidian length of a point is zero. Therefore testing whether parabolas are solutions of the system \eqref{eq:2d_case_affine} is meaningless. Affine curvature \cite{Shirokovy_1959}
      \begin{equation}\label{eq:affine_curvature}
            \kappa_{a} = \kappa_{ea}^{-3/2} \cdot d\kappa_{ea} / d\sigma_{ea}
      \end{equation}
       of a parabola is not defined.
      \vspace{0.2cm}
      \par\textbf{Non-ambiguity in sequences of affine transformations for piece-wise parabolic representation of ``drawing'' patterns.} Equi-affine curvature \eqref{eq:ea_curvature} is zero for parabolas, positive constant for ellipses including circles, and negative constant for hyperbolas \cite{Shirokovy_1959}. The main theorem of the equi-affine theory of plane curves states that ``the natural equation $\kappa_{ea} = f(\sigma_{ea})$ defines a plane curve up to an arbitrary equi-affine transformation" \cite{Shirokovy_1959}. Therefore given two arbitrary parabolic segments having the same equi-affine arc and noting that their equi-affine curvatures are identical, one segment can be transformed into the other by applying a unique equi-affine transformation \cite{Polyakov:2006, Polyakov_et_al_B.Cyb:2009}.
      \par Consequently, given two arbitrary parabolic segments with prescribed initial and final points, there exists a unique affine transformation \eqref{eq:affine_transformation} mapping one segment to the other\footnote{Already without requirement for equality of their equi-affine arcs as in case of equi-affine transformation.} in such a way that initial and final points are matched. However when direction of drawing is not prescribed there exist two affine transformations \eqref{eq:affine_transformation} mapping one parabolic segment to the other. The determinants of the linear part of the two transformations have opposite signs and the same absolute value, say $|w|$.
      \par Specifically, say parabolic segment $P_1$ connects point $A_1$ to $B_1$ and parabolic segment $P_2$, belonging to the same plane, connects point $A_2$ to $B_2$, as depicted in Figure \ref{fig:parabolas}. There exists unique affine transformation mapping $P_1$ into $P_2$ in such a way that $A_1$ is mapped into $A_2$ (correspondingly $B_1$ into $B_2$). There also exists unique affine transformation mapping $P_1$ into $P_2$ and $A_1$ into $B_2$ (correspondingly $B_1$ into $A_2$). Note also that $|w| = \left|\sigma_{ea, 2} / \sigma_{ea, 1}\right|^{3/2}$, $\sigma_{ea, i}$ is equi-affine arc \eqref{eq:ea_arc} of the $i$th segment. This property of existence and uniqueness of affine transformation relating two oriented parabolic segments follows from the main theorem of the equi-affine theory of plane curves mentioned above and possibility to represent any affine transformation with unique superposition of uniform scaling (either with or without reflection) and equi-affine transformation:
      \begin{equation}\label{eq:affine_is_equi_aff_and_scaling}
        T_{\mbox{affine}} = T_{\mbox{equi-affine}} \circ T_{\mbox{scaling}}\, .
      \end{equation}
      Uniform scaling without reflection is defined as 2x2 identity matrix multiplied by the scaling coefficient:
      \begin{equation}\label{eq:scaling_transformation_no_reflection}
            T_{\mbox{scaling}}^{\mbox{no reflection}} = |w| \cdot \left(
        \begin{array}{cc}
        1 & 0 \\
        0 & 1
     \end{array} \right)\,.
     \end{equation}
     Uniform scaling with reflection is defined as reflection transformation multiplied by the scaling coefficient:
     \begin{equation}\label{eq:scaling_transformation_with_reflection}
        T_{\mbox{scaling}}^{\mbox{with reflection}} = |w| \cdot \left(
\begin{array}{cc}
        1 & 0 \\
        0 & -1
     \end{array} \right)\,.
     \end{equation}
\par Parabolic segments connected into a sequence always have well identified initial and final points. Sequences of parabolic-like components revealed in monkey drawing movements \cite{Polyakov:2006, Polyakov_et_al_B.Cyb:2009, Polyakov_et_al_PLoS_C_B:2009} were usually implemented with unchanged direction of motion. Therefore piece-wise parabolic representation of complex patterns with sequences of affine transformations applied to a single (arbitrary!) parabolic template is unambiguous.
\par Decomposition of affine transformation \eqref{eq:affine_is_equi_aff_and_scaling} into sequence of two transformations, one of them scaling, may also take place in processing of spatial information in the brain. In particular, study \cite{Seculer.Nash:1972} demonstrated that a pair of mental transformations, size scaling and rotation, produced additive effects on reaction time, consistent with serial processing of these transformations. Rotations constitute a subgroup of Euclidian transformations that in turn form a subgroup of equi-affine transformations.

  \item \textbf{Circle}. Noting that affine arc is integrated square root of equi-affine curvature \eqref{eq:ea_curvature} ($\sigma_{a} = \int\sqrt{\kappa_{ea}} d\sigma_{ea}$ \cite{Shirokovy_1959}) and that equi-affine curvature of a circle is positive constant one immediately obtains for a circle: $\sigma_{ea} = \kappa_{ea}^{-1/2} \sigma_a$. Substituting into \eqref{eq:circle_equi_affine} one gets:
        \begin{equation}\label{eq:circel_affine}
        \begin{array}{l}
        x = x_0 + \kappa_{ea}^{-3/4} \cdot \cos(\sigma_{a}) \\
        y = y_0 + \kappa_{ea}^{-3/4} \cdot \sin(\sigma_{a})
        \end{array}
  \end{equation}
  which is solution of the system \eqref{eq:2d_case_affine}. Circles constitute non-invariant solutions under arbitrary affine transformation for arbitrary $n \geq 1$. The class of circles is invariant under similarity transformations \eqref{eq:similarity_transformation} and scalings with reflection \eqref{eq:scaling_transformation_with_reflection}.
  \item \textbf{Logarithmic spiral}. The speed of accumulating affine arc of the logarithmic spiral \eqref{eq:log_spiral_polar} with respect to changing polar angle $\varphi$ is the following constant \cite{Bright:2006}:
      \begin{equation*}
        d{\sigma}_a / d\varphi = \displaystyle\frac{\sqrt{9 + \beta^2}}{3}\; .
      \end{equation*}
      So setting $\varphi(0) = 0$,
      \begin{equation*}
            \varphi(\sigma_{a}) = \displaystyle\frac{3}{\sqrt{9 + \beta^2}} \sigma_{a}\; .
      \end{equation*}
      The expression for the logarithmic spiral \eqref{eq:log_spiral_polar} becomes
    \begin{equation}\label{eq:log_spiral_affine}
        \begin{array}{l}
        x(\sigma_{a}) = x_{0} + \mathrm{const} \cdot \exp{\left(\beta \frac{3}{\sqrt{9 + \beta^2}} \sigma_{a}\right) } \cos{\left(\frac{3}{\sqrt{9 + \beta^2}} \sigma_{a}\right)}\\\\
        y(\sigma_{a}) = y_{0} + \mathrm{const} \cdot \exp{\left(\beta \frac{3}{\sqrt{9 + \beta^2}} \sigma_{a}\right) } \sin{\left(\frac{3}{\sqrt{9 + \beta^2}} \sigma_{a}\right)}\; .
         \end{array}
     \end{equation}
     The values of $\beta$ for which logarithmic spiral satisfies the system \eqref{eq:2d_case_affine} depend on the degree of smoothness $n$ as follows.
     \begin{enumerate}
        \item $n = 1$: no solutions.
        \item $n = 2$: $\beta = \pm\sqrt{3}$ or $\beta \approx \pm 1.73205$.
        \item $n = 3$: $\beta = \pm\sqrt{5 \pm 2\sqrt{5}}$ (see also \cite{Bright:2006, Meirovitch:2014}) or $\beta \approx \pm 0.726543,\,  \pm 3.07768$.
        \item $n = 4$: $\beta^6 - 21\beta^4 + 35\beta^2 = 7$ implying\footnote{ Of course, exact solutions with roots can be written for this cubic equation.} $\beta \approx \pm 0.481575,\, \pm 1.25396,\, \pm 4.38129$.
        \item $n = 5$: $\beta^6 - 33 \beta^4 + 27\beta^2 = 3$ implying $\beta \approx \pm 0.36397,\, \pm 0.8391,\, \pm 5.67128$.
     \end{enumerate}
\end{enumerate}

 \begin{figure}% [h!]
 \vspace*{-1.5cm}
% \hspace*{-0.5cm}
\psfig{figure=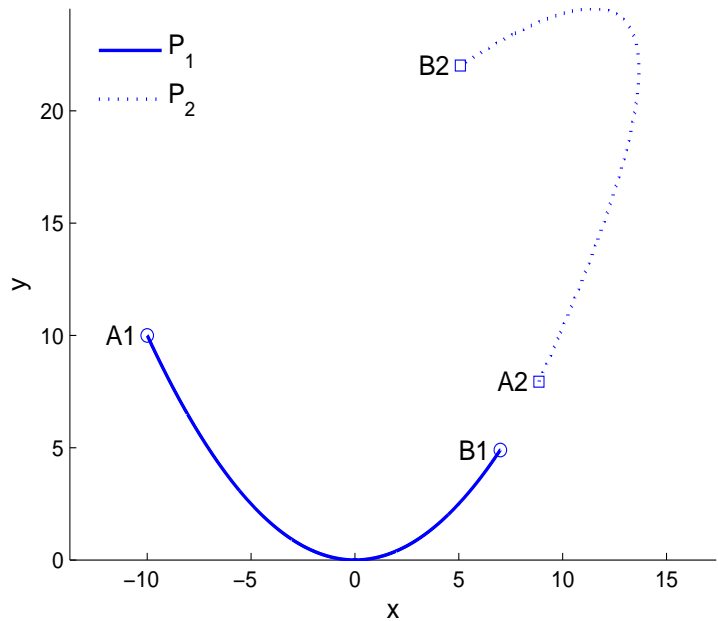,
height=85mm,width=105mm} %
\caption{\small \textbf{Two oriented parabolas.} Parabolic segment $P_1$ connects point $A_1$ to $B_1$ and is "drawn" counter-clock-wise. Parabolic segment $P_2$ connects point $A_2$ to $B_2$ and is also "drawn" counter-clock-wise. Affine transformation \eqref{eq:affine_transformation} that maps $P_1$ to $P_2$ and preserves the counter-clock-wise orientation has positive determinant of its linear part and transfers point $A_1$ into point $A_2$, correspondingly $B_1$ into $B_2$. There also exists a unique affine transformation that maps $P_1$ into $P_2$ while transferring $A_1$ into $B_2$, correspondingly $B_2$ into $A_1$. The determinant of the linear part of such transformation is negative, orientation of "drawing" flips from counter-clock-wise ($A_2$ towards $B_2$) into clock-wise ($B_2$ towards $A_2$).}\label{fig:parabolas}
      \end{figure}

\subsubsection*{Center-affine group}
\par Center-affine transformations of coordinates involve 4 independent parameters and are of the form:
\begin{equation}\label{eq:center_affine_transformation}
  \begin{array}{c}
      x_1 = \alpha x + \beta y \\
      y_1 = \gamma x + \delta y\, , \
  \end{array}
  \quad \left|
     \begin{array}{cc}
        \alpha & \beta \\
        \gamma & \delta
     \end{array}
  \right| \neq 0\, .
\end{equation}
The speed of accumulating center-affine arc is computed as follows \cite{Shirokovy_1959}:
\begin{equation}\label{eq:center_affine_speed}
    \dot{\sigma}_{ca} = \sqrt{\epsilon \, \frac{{\dot{\sigma}_{ea}}^3}{\left |
     \begin{array}{cc}
        x & \dot{x} \\
        y & \dot{y}
     \end{array}
  \right |}}\; ,
\end{equation}
with $\epsilon = \mbox{sign}\left[(x' y'' - y' x'') \cdot (x y' - y x')\right]$.
System \eqref{eq:general_form_higher_order} in plane  becomes
\begin{equation}\label{eq:2d_case_centeraffine}
    \left\{
    \begin{array}{l}
         x' x^{(2n)} +  y' y^{(2n)}  = 0  \\ \\
        \left(\displaystyle\frac{x' y'' - x'' y'}{x y' - x' y}\right) ^ 2 = 1 \, .
    \end{array}
    \right.
\end{equation}
Here concrete curves are considered as candidate solutions, some of them are filtered out for the prescribed center-affine parametrization (Approach \ref{lbl:approaches_candidate_solutions2}). The lower equation is automatically satisfied when a curve is parameterized with desired geometric measurement or when substitution \eqref{eq:formulae_substituttions} is used for such measurement.
The results for candidate solutions (parabola, circle, logarithmic spiral) are as follows:
\begin{enumerate}
  \item \textbf{Parabola}. As mentioned above, drawing parabola with constant equi-affine velocity does minimize the cost functional for $n \geq 2$. Drawing parabola with constant equi-affine velocity results in non-constant denominator in the expression of center-affine speed \eqref{eq:center_affine_speed}\footnote{Take $\sigma_{ea}(t) = t$ and substitute vector function describing a parabola, $x = \sigma_{ea},\; y = {\sigma_{ea}}^2 / 2$ into the denominator of \eqref{eq:center_affine_speed}.}. Therefore expression \eqref{eq:center_affine_speed} implies that the center-affine speed of the optimal trajectory along parabola is not constant and so the upper equation of the system \eqref{eq:2d_case_centeraffine} is not satisfied.
  \item \textbf{Circle}. Using parametrization of a circle with equi-affine arc as in \eqref{eq:circle_equi_affine}, equation \eqref{eq:center_affine_speed} implies that for a circle whose center coincides with the origin the center-affine speed is related to the equi-affine velocity as follows:
      $$\dot{\sigma}_{ca} = {\kappa_{ea}}^{1/8}\cdot{\dot{\sigma}_{ea}}\; .$$
      Equi-affine curvature of a circle is positive constant. Therefore drawing a circle with constant center-affine speed is equivalent to drawing a circle with constant equi-affine velocity. So system \eqref{eq:2d_case_centeraffine} is satisfied for any circle centered at the origin and for any $n \geq 1$ as circles satisfy the upper equation for the equi-affine parametrization.
  \item \textbf{Logarithmic spiral}. Direct computation based on \eqref{eq:center_affine_speed} implies that the speed of accumulating center-affine arc with respect to changing polar angle $\varphi$ from \eqref{eq:log_spiral_polar} is the following constant
      $$d{\sigma}_{ca} / d\varphi =\sqrt{1 + \beta^2}\;.$$
      Assuming the curve is centered at the origin (for every point on the curve, the angle between tangent to the curve at that point and the vector connecting the point to the origin is the same) and $\varphi(0) = 0$,
      $$\varphi(\sigma_{ca}) = \displaystyle\sqrt{\frac{1}{{1 + \beta^2}}} \sigma_{ca}\;.$$
      The expression for the logarithmic spiral \eqref{eq:log_spiral_polar} becomes
    \begin{equation}\label{eq:log_spiral_center_affine}
        \begin{array}{l}
        x(\sigma_{ca}) = \mathrm{const} \cdot \exp{\left(\beta \sqrt{\frac{1}{{1 + \beta^2}}} \sigma_{ca}\right) } \cos{\left(\sqrt{\frac{1}{{1 + \beta^2}}} \sigma_{ca}\right)}\\
        y(\sigma_{ca}) = \mathrm{const} \cdot \exp{\left(\beta \sqrt{\frac{1}{{1 + \beta^2}}} \sigma_{ca}\right) } \sin{\left(\sqrt{\frac{1}{{1 + \beta^2}}} \sigma_{ca}\right)}\; .
         \end{array}
     \end{equation}
     Logarithmic spiral satisfies the system \eqref{eq:2d_case_centeraffine} while corresponding values of $\beta$ depend on the degree of smoothness. The dependence is the same as in the case of affine arc!
\end{enumerate}

\subsubsection*{Equi-center-affine group}
\par Planar equi-center-affine transformations of coordinates involve 3 independent parameters and are of the form:
\begin{equation}\label{eq:equi_center_affine_transformation}
  \begin{array}{c}
      x_1 = \alpha x + \beta y \\
      y_1 = \gamma x + \delta y\, , \
  \end{array}
  \quad \left|
     \begin{array}{cc}
        \alpha & \beta \\
        \gamma & \delta
     \end{array}
  \right| = 1\, .
\end{equation}
The speed of accumulating equi-center-affine arc is computed as follows \cite{Shirokovy_1959}:
\begin{equation}\label{eq:equi_center_affine_speed}
    \dot{\sigma}_{eca} = \left |
     \begin{array}{cc}
        x & \dot{x} \\
        y & \dot{y}
     \end{array}
  \right |\; .
\end{equation}
System \eqref{eq:general_form_higher_order} in plane becomes
\begin{equation}\label{eq:2d_case_equicenteraffine}
    \left\{
    \begin{array}{l}
         x' x^{(2n)} +  y' y^{(2n)}  = 0  \\ \\
        x y' - x' y = 1 \, .
    \end{array}
    \right.
\end{equation}
Here concrete curves are considered as candidate solutions, some of them are filtered out for the prescribed equi-center-affine parametrization (Approach \ref{lbl:approaches_candidate_solutions2}). The lower equation is automatically satisfied when a curve is parameterized with desired geometric measurement or when substitution \eqref{eq:formulae_substituttions} is used for such measurement.
The results for candidate solutions (parabola, circle, logarithmic spiral) are as follows:
\begin{enumerate}
  \item \textbf{Parabola}. Note that the square root of the equi-center-affine speed is exactly the denominator in expression \eqref{eq:center_affine_speed} for the center-affine speed and is non-constant when the equi-affine speed is constant. Therefore optimal drawing of parabola has non-constant equi-center-affine speed, more detailed explanation is provided in case of the center-affine arc.
  \item \textbf{Circle}. Again, note that the square root of the equi-center-affine speed is equal to denominator of the center-affine speed \eqref{eq:center_affine_speed} while the numerator is a function of the equi-affine speed. Given that both the equi-affine and center-affine speeds are constant for optimal drawing of a circle centered at the origin, conclude that the equi-center-affine speed is also constant for this motion.
  \item \textbf{Logarithmic spiral}. Direct computation based on \eqref{eq:equi_center_affine_speed} implies that the speed of accumulating equi-center-affine arc with respect to changing polar angle $\varphi$ from \eqref{eq:log_spiral_polar} is the following constant
      $$d{\sigma}_{eca} / d\varphi = \mathrm{const} \cdot \exp\left(2 \beta \varphi \right)\;.$$
      Assuming the curve is centered at the origin as explained in the part related to the center-affine arc and $\varphi(0) = 0$,
      \begin{equation}\label{eq:log_spiral_phi_vs_sigma_eca}
            \varphi(\sigma_{eca}) = \displaystyle\frac{1}{2 \beta} \ln\left( \frac{2\, \beta\, \sigma_{eac}}{\mathrm{const}} + 1 \right)\;.
      \end{equation}
      The expression for $x$ and $y$ coordinates of the logarithmic spiral can be obtained by substituting \eqref{eq:log_spiral_phi_vs_sigma_eca} into \eqref{eq:log_spiral_polar} and setting $x_0 = y_0 = 0$.

      The values of $\beta$ for which logarithmic spiral parameterized with equi-center-affine arc satisfies the upper equation of system \eqref{eq:2d_case_equicenteraffine} depend on the degree of smoothness $n$ as follows.
     \begin{enumerate}
        \item $n = 1$: no solutions.
        \item $n = 2$: $\beta = \pm\sqrt{0.6}$.
        \item $n = 3$: $\beta = \pm\sqrt{\displaystyle\frac{1}{189} \left( 95 - 4 \sqrt{505} \right)} \approx \pm0.164448$.
        \item $n = 4$: $19305 \beta^6  - 25333 \beta^4 + 1435 \beta^2 - 7$ implying\footnote{Expressions with roots can be written for solutions of this cubic equation.} \\ $\beta \approx \pm 0.0734067,\, \pm 1.11945,\, \pm 0.231726$.
        \item $n = 5$: $ 1276275 \beta^8 - 1992932 \beta^6 + 169442 \beta^4  3 - 1988 \beta^2 + 3 = 0$ implying\footnote{Expressions with roots can be written for solutions of this quartic equation.} $\beta \approx \pm 0.0420802 ,\, \pm 0.109069 ,\, \pm 0.275325 ,\, \pm 1.21328 $.
     \end{enumerate}

\end{enumerate}

\subsubsection*{Euclidian group}\label{subsection:Euclidian_group}
\par Euclidian transformations of coordinates are 3-parametric, they are of the form:
\begin{equation}\label{eq:Eucl_transformation}
\begin{array}{c}
      x_1 =  x \cos\theta - y \sin\theta + a \\
      y_1 = x \sin\theta + y \cos\theta + b\, , \
  \end{array}
\end{equation}
The speed of accumulating Euclidian arc which is a standard notion of tangential speed whose integral is equal to the length of the drawn trajectory is computed as follows:
\begin{equation}\label{eq:Euclidian_speed}
    \dot{\sigma}_{eu} = \sqrt{\dot{x}^2 + \dot{y}^2}\; .
\end{equation}
System \eqref{eq:general_form_higher_order} in plane becomes
\begin{equation}\label{eq:2d_case_Euclidian}
    \left\{
    \begin{array}{l}
         x' x^{(2n)} +  y' y^{(2n)}  = 0  \\
         \\
        \displaystyle\sqrt{{x'}^2 + {y'}^2} = 1 \, .
    \end{array}
    \right.
\end{equation}
Here concrete curves are considered as candidate solutions, some of them are filtered out for the prescribed Euclidian parametrization (Approach \ref{lbl:approaches_candidate_solutions2}). The lower equation is automatically satisfied when a curve is parameterized with desired geometric measurement or when substitution \eqref{eq:formulae_substituttions} is used for such measurement.
The results for candidate solutions (parabola, circle, logarithmic spiral) are as follows:
\begin{enumerate}
  \item \textbf{Parabola}. Drawing parabola with constant equi-affine velocity does minimize the jerk and has non-constant Euclidian speed. Therefore the upper equation of the system \eqref{eq:2d_case_Euclidian} is not satisfied for parabolas parameterized with Euclidian arc.
  \item \textbf{Circle}. Motion with constant Euclidian speed along a circle is equivalent to motion with constant angular speed, therefore it is also equivalent to motion with constant equi-affine velocity. So the system \eqref{eq:2d_case_Euclidian} is satisfied for circles.
  \item \textbf{Logarithmic spiral}. Direct computation based on \eqref{eq:log_spiral_polar} and \eqref{eq:Euclidian_speed} implies that the speed of accumulating Euclidian arc-length with respect to changing polar angle $\varphi$ is the following expression
      $$d{\sigma}_{eu} / d\varphi =\mathrm{const} \cdot \sqrt{1 + \beta^2} \cdot e^{\beta\varphi}\; .$$
      After setting $\varphi(0) = 0$,
      $$\varphi(\sigma_{eu}) = \displaystyle\ln\left(\frac{\beta}{\mathrm{const} \cdot \sqrt{1 + \beta^2}}\sigma_{eu} + 1\right) / \beta\; .$$
       The expression for the logarithmic spiral \eqref{eq:log_spiral_polar}, up to Euclidian transformations, becomes
    \begin{equation}\label{eq:log_spiral_euclidian}
        \begin{array}{l}
        x(\sigma_{eu}) = \left(\frac{\beta}{\sqrt{1 + \beta^2}}\sigma_{eu} + \mathrm{const}\right) \cos{\left[ \ln\left(\frac{\beta}{\mathrm{const} \sqrt{1 + \beta^2}}\sigma_{eu} + 1\right) / \beta \right]}\\ \\
        y(\sigma_{eu}) = \left(\frac{\beta}{\sqrt{1 + \beta^2}}\sigma_{eu} + \mathrm{const}\right) \sin{\left[ \ln\left(\frac{\beta}{\mathrm{const} \sqrt{1 + \beta^2}}\sigma_{eu} + 1\right) / \beta \right]}\; .
         \end{array}
     \end{equation}
     Logarithmic spiral satisfies the system \eqref{eq:2d_case_Euclidian} when the value of $\beta$ is chosen appropriately for each value of the degree of smoothness $n$. For $n \leq 5$ the values of $\beta$ are as follows.
     \begin{enumerate}
        \item $n = 1$: $\beta$ is arbitrary. However the case of $n = 1$ does not actually correspond to a practical optimization problem because the corresponding cost function obtains the same values for any rule $\sigma(t)$, either having constant derivative or not.
        \item $n = 2$: no solutions.
        \item $n = 3$: $\beta = \pm\frac{1}{\sqrt{5}}$ (see also \cite{Bright:2006, Meirovitch:2014}) or $\beta \approx \pm 0.447214$.
        \item $n = 4$: $84 \beta^4 - 35 \beta^2 + 1 = 0$ implying $\beta = \pm\sqrt{\frac{35 \pm \sqrt{889}}{168}}$ or $\beta \approx \pm 0.17566,\,  \pm 0.621136$.
        \item $n = 5$: $3044 \beta^6 - 1869 \beta^4 + 126 \beta^2 - 1 = 0$ implying $\beta \approx \pm 0.095726,\, \pm 0.258088,\, \pm 0.733636$.
     \end{enumerate}
    These values of $\beta$ are different from the values of $\beta$ in cases of equi-affine, affine or equi-center-affine arcs.
\end{enumerate}

\subsubsection*{Similarity group}
\par Similarity transformations of coordinates involve 4 independent parameters and are of the form:
\begin{equation}\label{eq:similarity_transformation}
  \begin{array}{c}
      x_1 = \rho \left(x \cos\theta + y \sin\theta\right) + a \\
      y_1 = \rho \left(- x \sin\theta + y \cos\theta\right) + b \, , \
  \end{array}
\end{equation}
The similarity group consists of combined parallel translations and rotations (both forming Euclidian group), and uniform scaling defined by the parameter $\rho$, without reflection. The speed of accumulating the arc in the similarity group is computed as follows \cite{Shirokovy_1959}:
\begin{equation}\label{eq:similarity_speed}
    \dot{\sigma}_{si} = \frac{{\dot{\sigma}_{ea}}^3}{{\dot{\sigma}_{eu}}^2}\; .
\end{equation}
System \eqref{eq:general_form_higher_order} in plane  becomes
\begin{equation}\label{eq:2d_case_similarity}
    \left\{
    \begin{array}{l}
         x' x^{(2n)} +  y' y^{(2n)}  = 0  \\ \\
        \displaystyle\frac{x' y'' - x'' y'}{{x'}^2 + {y'}^2} = 1 \, .
    \end{array}
    \right.
\end{equation}
Here concrete curves are considered as candidate solutions, some of them are filtered out for the prescribed similarity parametrization (Approach \ref{lbl:approaches_candidate_solutions2}). The lower equation is automatically satisfied when a curve is parameterized with desired geometric measurement or when substitution \eqref{eq:formulae_substituttions} is used for such measurement.
The results for candidate solutions (parabola, circle, logarithmic spiral) are as follows:
\begin{enumerate}
  \item \textbf{Parabola}. Equation \eqref{eq:similarity_speed} implies that coordinates of a parabola parameterized with the similarity arc, up to a similarity transformation, are:
     \begin{equation*}
        \begin{array}{l}
        x(\sigma_{si}) = \mathrm{const} \cdot \tan\left( \sigma_{si}\right) \\
        y(\sigma_{si}) = \frac{\mathrm{const}}{2} \tan^2\left(\sigma_{si}\right)\; .
        \end{array}
     \end{equation*}
       As in the case of Euclidian arc, noting that maximally smooth drawing of parabola has constant equi-affine velocity and at the same time non-constant rate of accumulating Euclidian and similarity arcs, conclude that parabolas parameterized with similarity arc do not satisfy the upper equation of the system \eqref{eq:2d_case_similarity}.
  \item \textbf{Circle}. Drawing a circle with constant equi-affine velocity is equivalent to drawing the circle with constant Euclidian speed. Therefore formula \eqref{eq:similarity_speed} implies for circles that constancy of the similarity speed is also equivalent to the constancy of the equi-affine velocity. So system \eqref{eq:2d_case_similarity} is satisfied for any circle as circles satisfy the corresponding system of equations for the equi-affine parametrization.
  \item \textbf{Logarithmic spiral}. Direct computation based on \eqref{eq:similarity_speed} implies that the speed of accumulating similarity arc with respect to changing polar angle $\varphi$ from \eqref{eq:log_spiral_polar} is equal to 1:
      \begin{equation*}
        d{\sigma}_{si} / d\varphi = 1\; .
      \end{equation*}
      Setting $\varphi(0) = 0$, obtain
      \begin{equation*}
            \varphi = \sigma_{si}\; .
      \end{equation*}
      For $n \leq 5$ logarithmic spiral satisfies the system \eqref{eq:2d_case_similarity} for the same values of $\beta(n)$ as in the case of affine and center-affine groups!
\end{enumerate}

\subsubsection*{Parametrization with polar angle}
\par Sufficiently short segments of any smooth curve with non-zero curvature can be appropriately translated and rotated so that they can be further parameterized with polar angle, for example part of a circle translated so that circle's center coincides with the origin. Plane curve parameterized with polar angle is represented in polar coordinates $(\rho,\, \theta),\, \rho \geq 0$ as follows:
\begin{equation*}% \label{eq:2d_case_Euclidian}
    \begin{array}{l}
         x(\theta) = \rho(\theta) \cos(\theta)  \\
         y(\theta) = \rho(\theta) \sin(\theta)\, .
    \end{array}
\end{equation*}
Note that for an arbitrary parametrization $\theta(t)$
\begin{equation*}% \label{eq:2d_case_Euclidian}
    \dot{\theta} = \frac{ \dot{x} y - \dot{y} x }{ x^2 + y^2 }\, .
\end{equation*}
So, for parametrization with polar angle the system \eqref{eq:general_form_higher_order} in plane becomes
\begin{equation}\label{eq:polar_angle}
    \left\{
    \begin{array}{l}
         x' x^{(2n)} +  y' y^{(2n)}  = 0  \\
         \\
        \displaystyle\frac{ x' y - y' x }{ x^2 + y^2 }  =  \displaystyle\frac{\sigma_{eca}}{  x^2 + y^2} = 1\, .
    \end{array}
    \right.
\end{equation}
Here concrete curves are considered as candidate solutions, some of them are filtered out for the prescribed parametrization with polar angle (Approach \ref{lbl:approaches_candidate_solutions2}). The lower equation is automatically satisfied when a curve is parameterized with desired geometric measurement or when substitution \eqref{eq:formulae_substituttions} is used for such measurement.
\par Parabolas are not solutions of the system \eqref{eq:polar_angle} while circles are and logarithmic spirals (translated to be centered at the origin\footnote{For every point on a circle or logarithmic spiral, the angle between tangent to the curve at that point and the vector connecting the point to the origin is the same.}) are for the same values of $\beta$ as in case of affine arc.
\par I call \textbf{"cosh spiral", "sinh spiral"} the following curves
\begin{equation}\label{eq:vectorial_expression_cosh_not_solution}
     \begin{array}{l}
        x(\sigma) = \cos(\sigma) \cosh(\beta \sigma) \\
        y(\sigma) = \sin(\sigma) \cosh(\beta \sigma)\; ,
     \end{array}
\end{equation}
\begin{equation}\label{eq:vectorial_expression_sinh_not_solution}
     \begin{array}{l}
        x(\sigma) = \cos(\sigma) \sinh(\beta \sigma) \\
        y(\sigma) = \sin(\sigma) \sinh(\beta \sigma)\; .
     \end{array}
\end{equation}
The coordinates of both curves are equal to the sum/difference of coordinates of two logarithmic spirals with values of $\beta$ having opposite signs\footnote{$\cosh(\beta \sigma) = \frac{1}{2}\left( e^{\beta \sigma} + e^{-\beta \sigma} \right)$, $\sinh(\beta \sigma) = \frac{1}{2}\left( e^{\beta \sigma} - e^{-\beta \sigma} \right)$.}. Neither of the two curves satisfies the upper equation in the systems \eqref{eq:general_form_lower_order}, \eqref{eq:general_form_higher_order} when parameterized by any of six above-mentioned arcs (equi-affine, affine, center-affine, equi-center-afine, Euclidian, similarity). Nevertheless both curves satisfy system \eqref{eq:polar_angle} for parametrization with polar angle when values $\beta(n)$ for $n \leq 5$ are the same as in case of logarithmic spiral parameterized with affine arc (the values of $\beta(n)$ can be found earlier in text). Note also that the curves in equations \eqref{eq:vectorial_expression_cosh_not_solution}, \eqref{eq:vectorial_expression_sinh_not_solution} have been translated so that the logarithmic spirals whose coordinates are summed/subtracted are centered at the origin. The curves in \eqref{eq:vectorial_expression_cosh_not_solution}, \eqref{eq:vectorial_expression_sinh_not_solution} are invariant solutions of \eqref{eq:polar_angle} under rotations and reflections. The "cosh" and "sinh" spirals and logarithmic spiral, all three with the same value of $\beta = \sqrt{5 - 2 \sqrt{5}}$ are depicted in Figure\ref{fig:spirals}.
\begin{figure}% [h!]
 % \vspace*{-0.5cm}
%\hspace*{-1.5cm}
    \psfig{figure=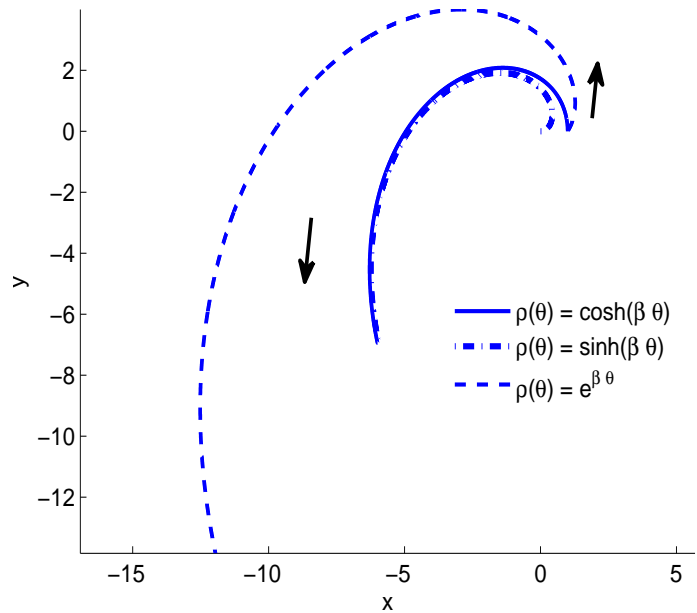,
    height=85mm,width=100mm} %
    \caption{\small\textbf{Three spirals.} "cosh spiral" \eqref{eq:vectorial_expression_cosh_not_solution} (solid), "sinh spiral" \eqref{eq:vectorial_expression_sinh_not_solution} (dots) and logarithmic spiral \eqref{eq:log_spiral_polar} (dashed), all three have $\beta = \sqrt{5 - 2 \sqrt{5}}$ that makes the curves solutions of the upper equation for parameterizations with polar angle (cosh and sinh spirals) and affine/center-affine/similarity arcs/polar angle (logarithmic spiral). Drawing is implemented counter-clock-wise, $0 \leq \theta \leq 4$. The spirals in the plot are identified with their equations in polar coordinates.}\label{fig:spirals}
\end{figure}
\subsubsection*{Pseudo solutions}
\par I call pseudo-solution a vector function that satisfies the upper equation of the systems \eqref{eq:general_form_lower_order}, \eqref{eq:general_form_higher_order} but does not satisfy the lower equation of the system. In other words parametrization of the curve provided by the vector function is not consistent with desired geometric parameter. Example \ref{example:parametrization} demonstrates that vector function \eqref{parabola_canonical_form_example} parameterizes parabola with equi-affine arc and is solution of the system \eqref{eq:general_form_higher_order} while considering the same vector function being parameterized with Euclidian arc results in pseudo solution.
\subsubsection*{Summary for known planar solutions}

\par The results for candidate planar solutions considered above are summarized in Table \ref{Table_mathematical_objects}. In particular, Table \ref{Table_mathematical_objects} shows that parabolas constitute affinely invariant solutions of the equations \eqref{eq:general_form_lower_order} and \eqref{eq:general_form_higher_order} for all degrees of smoothness above 1 in equi-affine geometry and are not solutions for any $n$ in any other geometry considered. Circles are solutions in all geometries for any degree of smoothness and possess invariance under composition of Euclidian transformations \eqref{eq:Eucl_transformation} and scaling \eqref{eq:scaling_transformation_no_reflection} including reflection \eqref{eq:scaling_transformation_with_reflection}. Logarithmic spirals were analyzed for up to 5th degree of smoothness. For degrees of smoothness 2-5 the parameter $\beta$ of the logarithmic spiral \eqref{eq:log_spiral_polar} depends on the degree of smoothness and the geometry in which the curve with such $\beta$ constitutes solution of the equation.  Circles and logarithmic spirals are analyzed separately, however a circle can be considered as a particular case of logarithmic spiral with $\beta = 0$. In case of center-affine and equi-center-affine geometries the curves are not invariant under translation and therefore should be translated appropriately before parameterizing with corresponding arcs; so in these two geometries I consider circles and logarithmic spirals centered at the origin and straight lines crossing the origin.
\par For a logarithmic spiral \eqref{eq:log_spiral_polar}, in the geometries analyzed in this work, the left hand side of the upper equation in the system \eqref{eq:general_form_lower_order} is equal to some function of $\beta$ (say $f(\beta)$) multiplied by a function depending exponentially on the polar angle\footnote{In other words, $\mbox{Left Hand Side} = f(\beta) \cdot exp(g(\beta) \cdot \varphi)$.}. Therefore only logarithmic spirals with $\beta$ nullifying $f(\beta)$ are solutions of \eqref{eq:general_form_lower_order}, \eqref{eq:general_form_higher_order} implying that the constant in the right hand side of the upper equation in \eqref{eq:general_form_lower_order}\footnote{First integral of the upper equation in the system \eqref{eq:general_form_higher_order}.} is zero for all solutions of \eqref{eq:general_form_lower_order} being logarithmic spirals in the geometries considered.
\begin{landscape}
\begin{table}[!ht]
\begin{adjustwidth}{0in}{0in}
%\vspace*{-0.5cm}
%%% \vspace*{-0cm}
\hspace*{-0.5cm}
% \begin{sideways}
{\small
\begin{tabular}{|l|l|c|c|c|c|c|c|c|}
\hline
Object  & Invariance for &  $n$ &  Equi-affine & Affine & Center-affine & Equi- & Euclidian & Polar \\
        & solutions        &   &  arc    & arc & (circle and log spiral& center- & arc & angle        \\
        & of \eqref{eq:general_form_higher_order} & & & & are centered & affine & &\\
        &   & & & & at the origin) & arc & &\\
        & & & & & \& similarity arcs & & &\\
\hline
\textbf{Parabola}& Affine & $\geq 2$ & Any         & Not relevant    & None       & None   & None     &  None \\
                  &         &     &                  &  (zero arc)         &      &     &  &      \\
    \hline

 \textbf{Straight}& {
  Affine} & {
   $\geq 1$ } & {
    Not relevant   }      & {
    Not relevant }   & {
    Not relevant  }      & Not    & {
    Any  }  & Not     \\
   {
    \textbf{line}     }          &         &     & {
      (zero arc)    }            & {
        (arc not defined)   }          & {
          (arc zero or    }     &   relevant   &     &  relevant   \\
   & & & & & {
   not defined) } & (zero) &     & (zero) \\
\hline
\textbf{Circle} & Uniform scaling &  $\geq 1$ &  Any         & Any    & Any           & Any   &  Any   &  Any  \\
                  & and Euclidian  &            &                 &               &           &     &     &    \\
\hline
\textbf{Logarithmic} &  & 2 & None  & $\beta \approx \pm 1.732$ & & $\beta = \pm\sqrt{0.6},$  & None &  \\
 \textbf{spiral} \eqref{eq:log_spiral_polar} &   & 3        &   $\beta \approx \pm 1.134$    & $\beta \approx \pm 0.727$,              &  &  $\beta \approx \pm 0.164$   &    $\beta \approx \pm 0.447$  &  \\
        &   Uniform scaling    &   &                               & $\beta \approx \pm 3.078$   &  Same  & & & Same\\
        & and Euclidian & 4 & $\beta \approx \pm 0.330$, & $\beta \approx \pm 0.482$, & & $\beta \approx \pm 0.073,$ & $\beta \approx \pm 0.176$, & \\
        &             &   & $\beta \approx \pm 1.580$   & $\beta \approx \pm 1.254$,  & as & $\beta \approx \pm 1.119,$ & $\beta \approx \pm 0.621$ & as\\
        &             &   &                               & $\beta \approx \pm 4.381$   & & $\beta \approx \pm 0.232$ &     &\\
        & & 5 & $\beta \approx \pm 0.167$, & $\beta \approx \pm 0.364$, & affine & $\beta \approx \pm 0.042,$  & $\beta \approx \pm 0.096$, & affine\\
        & &   & $\beta \approx \pm 0.480$, & $\beta \approx \pm 0.839$,& & $\beta \approx \pm 0.109,$ & $\beta \approx \pm 0.258$, & \\
        & &   & $\beta \approx \pm 1.878$, & $\beta \approx \pm 5.671$ & & $\beta \approx \pm 0.275,$ & $\beta \approx \pm 0.734$ & \\
        & &   &                            &                           & & $\beta \approx \pm 1.213$ &                           & \\
\hline
\textbf{``cosh spiral''}  & Uniform scaling &  &   None       & None     & None          &   & None      &  Same as \\
 \eqref{eq:vectorial_expression_cosh_not_solution}                  &       &      &                 &               &           &  &     & affine  \\
\textbf{``sinh spiral''}    & and Euclidian &  &           &      &            &  &     & for  log.   \\
 \eqref{eq:vectorial_expression_sinh_not_solution}                  &        &
      &                 &               &           &       &     &  spiral \\
\hline
\end{tabular}}
% \end{sideways}
\end{adjustwidth}
\caption{\small{\bf Known solutions and their invariance for different degrees of smoothness.}  Known solutions of the systems \eqref{eq:general_form_lower_order} and \eqref{eq:general_form_higher_order} in plane for different orders of trajectory smoothness $n$ and geometric parameterizations invariant in affine group and four of its subgroups. Invariance of the class of curves and value of the parameter $\beta$ of the logarithmic spiral are indicated. Curves parameterized with center-affine and equi-center-affine arcs and with polar angle are considered centered at the origin by parallel translation: straight lines cross the origin, tangents to logarithmic spirals and circles have constant angle with the vectors connecting corresponding points on the curves to the origin.}\label{Table_mathematical_objects}
\end{table}
\end{landscape}
\subsection*{Solutions in space for spatial curves parametered with 3D equi-affine arc}\label{section:3d_equi_affine}
\par Equi-affine transformations of coordinates in space involve 11 independent parameters and are of the form:
\begin{equation}\label{eq:3Dequi-affine_transformation}
  \begin{array}{c}
      x_1 = \alpha_{1\, 1} x + \alpha_{1\, 2} y + \alpha_{1\, 3} z + a \\
      y_1 = \alpha_{2\, 1} x + \alpha_{2\, 2} y + \alpha_{2\, 3} z + b \\
      z_1 = \alpha_{3\, 1} x + \alpha_{3\, 2} y + \alpha_{3\, 3} z + c , \
  \end{array}
  \quad \left |
     \begin{array}{ccc}
        \alpha_{1\, 1}  & \alpha_{1\, 1}  & \alpha_{1\, 3} \\
        \alpha_{2\, 1}  & \alpha_{2\, 1}  & \alpha_{2\, 3} \\
        \alpha_{3\, 1}  & \alpha_{3\, 1}  & \alpha_{3\, 3}
     \end{array}
  \right |  = 1\, .
\end{equation}
The speed of accumulating spatial equi-affine arc is computed as follows \cite{Shirokovy_1959}:
\begin{equation}\label{eq:ea3_speed}
    \dot{\sigma}_{ea3} = \left |
     \begin{array}{ccc}
        \dot{x} & \ddot{x} & \dddot{x} \\
        \dot{y} & \ddot{y} & \dddot{y} \\
        \dot{z} & \ddot{z} & \dddot{z} \\
     \end{array}
  \right | ^{1 / 6}
\end{equation}
and is called spatial equi-affine velocity. It has been proposed that 3-dimensional movements conserve spatial equi-affine velocity and the conservation phenomenon was called the ``1/6 power-law'' \cite{Pollick:2009, Maoz:2009}.
The system \eqref{eq:general_form_higher_order} becomes
\begin{equation}\label{eq:3d_case_equiaffine_higher_order}
    \left\{
    \begin{array}{l}
        x' x^{(2n)} + y' y^{(2n)} + z' z^{(2n)} = 0 \\
        \\
        \left |
     \begin{array}{ccc}
        x' & x'' & x''' \\
        y' & y'' & y''' \\
        z' & z'' & z''' \\
     \end{array}
  \right | = 1 \; .
    \end{array}
    \right.
\end{equation}
A list of all spatial curves, up to spatial equi-affine transformations, with constant spatial equi-affine curvature and torsion\footnote{Formula for the spatial equi-affine curvature\cite{Shirokovy_1959}: $\chi(\sigma_{ea3}) =        \left |
     \begin{array}{ccc}
        x' & x''' & x^{(4)} \\
        y' & y''' & y^{(4)} \\
        z' & z''' & z^{(4)} \\
     \end{array}\right|$.
     Formula for the equi-affine torsion\cite{Shirokovy_1959}: $\tau(\sigma_{ea3}) =        - \left |
     \begin{array}{ccc}
        x'' & x''' & x^{(4)} \\
        y'' & y''' & y^{(4)} \\
        z'' & z''' & z^{(4)} \\
     \end{array} \right|$. The differentiation is implemented with respect to the spatial equi-affine arc.
     } is provided in \cite{Shirokovy_1959}. Considering these curves as candidate solutions, some of them got filtered out for the prescribed parametrization (Approach \ref{lbl:approaches_candidate_solutions2}). The lower equation is automatically satisfied when a curve is parameterized with desired geometric measurement or when substitution \eqref{eq:formulae_substituttions} is used for such measurement. Only the following two curves from the list are solutions of the system \eqref{eq:3d_case_equiaffine_higher_order}.
\begin{enumerate}
  \item \textbf{Parabolic screw line} \cite{Polyakov:2006, Polyakov_et_al_B.Cyb:2009} can be parameterized with spatial equi-affine arc
    \begin{equation}\label{eq:ea3_arc}
        \sigma_{ea3} = \int_0^t \dot{\sigma}_{ea3}(\tau) d\tau\; .
    \end{equation}
      up to a spatial equi-affine transformation, as follows:
      \begin{equation}\label{eq:parabolic_screw_line}
        \begin{array}{rcl}
            x & = & \sigma_{ea3} \\
            y & = & {\sigma_{ea3}}^2 / 2 \\
            z & = & {\sigma_{ea3}}^3 / 6 \; .
        \end{array}
      \end{equation}
      Parabolic screw line is invariant solution under arbitrary spatial equi-affine transformations when $n \geq 2$. The class of parabolic screw lines is invariant under arbitrary spatial affine transformations.
  \item  \textbf{Elliptic screw line} can be parameterized with spatial equi-affine arc up to a spatial equi-affine transformation as follows:
        \begin{equation}\label{eq:elliptic_screw_line}
        \begin{array}{rcl}
            x & = & \mathrm{const} \cdot \cos\left(\mathrm{const}^{-1/3} \sigma_{ea3} \right) \\
            y & = & \mathrm{const} \cdot \sin\left(\mathrm{const}^{-1/3} \sigma_{ea3} \right) \\
            z & = & \mathrm{const}^{-1/3} \sigma_{ea3}\; .
        \end{array}
      \end{equation}
      Spatial Euclidian transformations composed with scaling and reflection transformations applied to elliptic screw lines result in solutions of the system \eqref{eq:3d_case_equiaffine_higher_order} as well for $n \geq 1$. However the cost functional with $n = 1$ obtains the same values for any rule $\sigma(t)$. Therefore the case of $n = 1$ is not interesting.  Arbitrary equi-affine transformations of the elliptic screw line of the form \eqref{eq:elliptic_screw_line} will not necessarily be solutions of the system \eqref{eq:3d_case_equiaffine_higher_order}.
\end{enumerate}
\vspace{0.3cm}
\par Both parabolic and elliptic screws have constant spatial equi-affine curvature (zero for the parabolic screw) and zero equi-affine torsion \cite{Shirokovy_1959}. Their Euclidian curvature and torsion are not zero. Elliptic screw has constant Euclidian curvature and torsion. Study \cite{Bright:2006} used another method to show that trajectories along the curves with constant Euclidian curvature and torsion simultaneously satisfy the constrained minimum-jerk model and the 1/6 power-law.
\par Parabolic screw line parameterized with Euclidian arc is not solution of the 3-dimensional versions of the upper equations in the systems \eqref{eq:general_form_lower_order}, \eqref{eq:general_form_higher_order}  while elliptic screw is.

\subsection*{Case of arbitrary parametrization, dimension and order of smoothness}\label{section:3d_equi_affine}
\par Equation \eqref{eq:only_upper_equation_higher_order} can be used to easily identify some curves being solutions for parameterizations whose properties are not prescribed in advance.
\begin{enumerate}
    \item Curves $\bm{r}_{L}(\sigma) = (r_1(\sigma),\, \ldots,\, r_L(\sigma))$ in $L$ dimensional space whose coordinates are polynomials of order $\leq 2 n - 1$ of an arbitrary parameter $\sigma$ obviously satisfy equation \eqref{eq:only_upper_equation_higher_order} as their $2n$-th order derivatives with respect to $\sigma$ are all zero:
        \begin{equation}\label{eq:solutions_polynomials}
            r_i(\sigma) = \sum_{k = 0}^{2 n - 1}a_{i, k} \sigma^k,\; i = 1,\, \ldots,\, L\, .
        \end{equation}
        Such solutions are invariant under arbitrary linear transformations and translations in $L$ dimensional space, including affine transformations. Parabolas constitute a particular case of \eqref{eq:solutions_polynomials} for $L = 2$ and curve's coordinates being particular 2-nd order polynomials with respect to $\sigma_{ea}$.
    \item In spaces of dimensions $L \geq 2$ take any two coordinates to be described as coordinates $x$ and $y$ of (a) logarithmic spiral \eqref{eq:log_spiral_polar} or of (b) "cosh spiral" \eqref{eq:vectorial_expression_cosh_not_solution}, or of (c) "sinh spiral" \eqref{eq:vectorial_expression_sinh_not_solution}. Whenever $L \geq 3$ describe the rest $L - 2$ coordinates with arbitrary polynomials of order $\leq 2 n - 1$. Obviously, left hand side of the equation \eqref{eq:only_upper_equation_higher_order} will be the same as the left hand side of the planar version of the equation for the curves (a)-(c) because $2n$th order derivatives of the polynomials vanish. Therefore the curves of the form (a) - (c) will satisfy equation \eqref{eq:only_upper_equation_higher_order} with the same values $\beta(n)$ as in case of parametrization with polar angle in plane presented earlier in text. So, without limitation of generality
        \begin{equation}\label{eq:solutions_spirals_polynomials}
            \begin{array}{rcl}
                r_1(\sigma) & = & x(\sigma)\, , \\
                r_2(\sigma) & = & y(\sigma)\, , \\
                r_i(\sigma) & = & \sum_{k = 0}^{2 n - 1}a_{i-2, k} \sigma^k,\; i = 3,\, \ldots,\, L\, ,
            \end{array}
        \end{equation}
        with $x(\sigma)$, $y(\sigma)$ being coordinates of the above-mentioned curves (a) - (c). Geometric parametrization of the curves (a) - (c) induced by $\sigma$ is invariant (still parameterizes solutions) when a curve undergoes Euclidian transformation and uniform scaling (including reflections) in $L$ dimensional space because such transformations preserve orthogonality.
\end{enumerate}
\section*{Discussion}
\par This work considers the problem of finding paths whose maximally smooth trajectories accumulate geometric measurement along the path with constant rate. The order of smoothness is arbitrary. Class of differential equations obeyed by such paths is derived. System of two differential equations corresponds to each order of smoothness $n$. Systems of special interest and their known solutions are presented. Detailed exposition with clarifying examples is provided for the rationale of the research question, derivation of the equations and their applications for different orders of smoothness, in different geometries (equi-affine and 5 other geometries) and dimensions. Derived class of equations constitutes generalization of earlier works on one hand and a tool for further developments in the field of motor control on the other hand, eg. finding more solutions and thus revealing new candidates for primitive shapes.
\par Earlier works proposed several curves for whom presumably the predictions of the constrained minimum-jerk model ($n = 3$) have constant rate of accumulating equi-affine velocity \cite{Polyakov:2001, Polyakov_etc_2001, Polyakov:2006, Bright:2006, Polyakov_et_al_B.Cyb:2009}, or planar Euclidian or affine velocity \cite{Bright:2006}, see also \cite{Meirovitch:2014}. Logarithmic spirals with such properties were proposed by Ido Bright\cite{Bright:2006}. Equations with values of $n$ different from 3 and solutions of the system of equations in arbitrary dimensional space $L$ is the main novelty of this work. Classes of curves described by formulas \eqref{eq:solutions_polynomials}, \eqref{eq:solutions_spirals_polynomials} and their particular cases \eqref{eq:vectorial_expression_cosh_not_solution}, \eqref{eq:vectorial_expression_sinh_not_solution} in plane have been proposed in the framework of invariance-smoothness criteria for movement trajectories for the first time.

%%%Among different solutions of the system \eqref{eq:general_form_higher_order} mentione Such solutions as parabolas, circles and parabolic screw appeared earlier in works \cite{Polyakov:2001, Polyakov_etc_2001, Polyakov:2006, Polyakov_et_al_B.Cyb:2009}. Logarithmic spirals as shapes mediating between the minimum-jerk and the 2/3 power-law models were proposed by Ido Bright and appeared earlier

% \footnote{Interestingly, logarithmic spirals constitute a general solution zeroing the second order differential invariant arising from the second prolongation of the generator of rotation group \cite{Ibragimov.book.azbuka:1989}. Rotations and parallel translations form Euclidian group of transformations.}

%%%in works \cite{Bright:2006, Polyakov_et_al_B.Cyb:2009}. Bright also proposed logarithmic spirals for cases of constant rate of accumulating Euclidian and affine arcs by maximally smooth movements with $n = 3$ \cite{Bright:2006}. Curves with constant Euclidian curvature and torsion (actually elliptic screw) were proposed by Bright as shapes along which maximally smooth trajectories satisfy the 1/6 power-law \cite{Bright:2006}. Curves of the form \eqref{eq:solutions_polynomials}, \eqref{eq:solutions_spirals_polynomials} and particular cases \eqref{eq:vectorial_expression_cosh_not_solution}, \eqref{eq:vectorial_expression_sinh_not_solution} are proposed as solutions of the derived equations for the first time.
%
%
%
\subsection*{Empirical rationale}
\par This work presents mathematical result of a largely predictive study. Nevertheless the incentive of deriving the equations for degrees of smoothness above 3 is based on several pieces of already existing empirical evidence related to importance of geometric invariance and kinematic smoothness in control of biological movements. Empirical findings revealed that geometric properties of biological movements dictate their speed, eg. the two-thirds power-law \cite{Lacquaniti_Terzuolo_Viviani_1983} which is equivalent to piece-wise constancy of movement's equi-affine velocity. Further studies of the geometric aspects of biological movements supported the validity of the two-thirds power-law and existence of the neural representation of the equi-affine velocity during movements' perception and performance \cite{Viviani_Stucchi_1992, Levit-Binnun.Schechtman.Flash:2006, Polyakov:2006, Dayan.etal.Flash:2007, Polyakov_et_al_B.Cyb:2009}.
\par During practice monkey scribbling movements became clustered into a small number of relatively long parabolic-like movement elements; analysis of motor cortical activity underlying monkey scribbling movements supported the hypothesis about parabolic movement primitives being represented in motor cortical activity synchronized to neural states \cite{Polyakov:2006, Polyakov_et_al_B.Cyb:2009, Polyakov_et_al_PLoS_C_B:2009}. Sosnik et. al. demonstrated that over the course of practice sequences of nearly straight point-to-point drawing movements by humans get coarticulated into smooth movements that can be well approximated with minimum-jerk trajectories passing through a single via-point \cite{Sosnik_etc_2004}. Such trajectories are parabolic-like \cite{Polyakov:2006, Shpigelmacher:2006,Polyakov_et_al_B.Cyb:2009}. Therefore convergence of non-smooth movements into smooth parabolic-like is natural for monkeys and humans. Parabolas constitute an equi-affine solution of the derived class of equations (that is for $\sigma = \sigma_{ea}$)  for an arbitrary degree of smoothness $n$ above 1 and not only for $n = 3$.
\par Equation \eqref{equation_equiaffine_example} was applied earlier to the parametrization with planar and spatial equi-affine arcs \cite{Polyakov:2006, Polyakov_et_al_B.Cyb:2009}. Further work \cite{Bennequin:2009, Fuchs:2010} analyzed geometric invariance and suggested that biological movements may be represented not only in either equi-affine or Euclidian geometries but simultaneously in multiple, and including affine, geometries while geometric representation guides temporal properties of movements. The predictions of the theory were tested on three data sets: drawings of elliptical curves, locomotion and drawing of complex figural forms. The authors claimed that their theory accounted well for the kinematic and temporal features of the movements, in most cases overperforming the constrained minimum jerk model. In addition, more recent findings show neural representation of scale invariance \cite{Kadmon.Flash:2014} (relevant for the similarity group \eqref{eq:similarity_transformation}). An important empirical evidence showed connection between the level of activation in different motor ares (M1, PMd, pre SMA) and the degree of motion smoothness acquired during learning to coarticulate point-to-point segments into complex smooth trajectories \cite{sosnik.flash.sterkin.hauptmann.karni:2014}. So there exists an empirical evidence that 1) invariance in different geometries and 2) level of hand trajectory's smoothness are represented in neural activity. Those two empirical characteristics of drawing-like movements and their mathematical properties motivated extension of the equation \eqref{equation_equiaffine_example} to arbitrary degree of smoothness and further demonstration of the method in different geometries here, including geometries defined by the similarity, center-affine and equi-center-affine groups. The two above-mentioned features being merged result in paths satisfying the derived systems of equations. The methodology demonstrated in the manuscript may be further applied for parameterizations not mentioned in this work.
\par Incorporation of movement primitives paradigm, type of geometric invariance, and level of trajectory's smoothness into methods for decoding neural data may provide additional information useful for the algorithms employed for brain-machine interfaces. Earlier works proposed neural network models for composing complex movements from primitives \cite{Schrader.Diesmann.Morrison:2011, Hanuschkin.Herrmann.Morrison.Diesmann:2011}.
\par Earlier empirical studies demonstrated existence of spontaneously generated hand trajectories with nearly straight\footnote{The end-points of the trajectories were prescribed but not trajectories' form.} \cite{Flash_Hogan_1985} and parabolic-like \cite{Polyakov:2006, Polyakov_et_al_B.Cyb:2009, Polyakov_et_al_PLoS_C_B:2009} paths. Empirical study of additional (to straight and parabolic segments) solutions of the derived equations as candidate movement primitives may lead to better understanding of movements' neural representation and production. At the moment I point to the following candidates: circles and logarithmic spirals (2 dimensions) and parabolic and elliptic screws (3 dimensions). Circles are actually a particular case of the logarithmic spiral \eqref{eq:log_spiral_polar} with $\beta = 0$.
\par A different study related to the research topic presented here reported that logarithmic spirals are affine orbits (have constant affine curvature \eqref{eq:affine_curvature}\footnote{So indeed, for logarithmic spiral \eqref{eq:log_spiral_polar} its affine curvature is equal to the following constant: $\kappa_{a} = 4 \beta / \sqrt{9 + \beta^2}$.}) and investigated logarithmic spirals in the framework of the mixed-geometry approach \cite{Meirovitch.Flash:2013, Meirovitch:2014}. Work \cite{Meirovitch:2014} focused on logarithmic spirals and other affine orbits and showed for each spiral to which mixture of geometries it belongs with respect to jerk minimization. An algorithm for representing movements based on mixed geometry approach and jerk optimization was proposed \cite{Meirovitch.Flash:2013}.
\par Formula \eqref{eq:affine_curvature} can be integrated to show that equi-affine curvature $\kappa_{ea}$ of any curve with constant affine curvature $\kappa_{a}$ is the following function of curve's equi-affine arc: $\kappa_{ea}(\sigma_{ea}) = 4 / \left(\kappa_{a} \sigma_{ea} + \mbox{const}\right)^2$. Therefore equation \eqref{eq:ea_curvature_scaling_factor} can be integrated\footnote{Following 2.14 of \cite{Kamke:1959},
(1)
if $\kappa_{a}^2 > 16$, $x'(\sigma_{ea}) = C_{1x} \left( \kappa_{a} \sigma_{ea} + \mbox{const} \right)^{1/2 + \sqrt{1 - (16 / \kappa_{a})}} + C_{2x} \left( \kappa_{a} \sigma_{ea} + \mbox{const} \right)^{1/2 - \sqrt{1 - (16 / \kappa_{a})}}$;
(2)
if $\kappa_{a}^2 = 16$, $x'(\sigma_{ea}) = C_{1x} \sqrt{\kappa_{a} \sigma_{ea} + \mbox{const}} +  C_{2x} \sqrt{\kappa_{a} \sigma_{ea} + \mbox{const}} \cdot \ln\left( \kappa_{a} \sigma_{ea} + \mbox{const} \right)$;
(3)
if $\kappa_{a}^2 < 16$, $x'(\sigma_{ea}) = C_{1x} \sqrt{\kappa_{a} \sigma_{ea} + \mbox{const}} \cdot \cos\left(0.5 \sqrt{16 / \kappa_{a} - 1} \cdot \ln\left(\kappa_{a} \sigma_{ea} + \mbox{const} \right) \right) + C_{2x} \sqrt{\kappa_{a} \sigma_{ea} + \mbox{const}} \cdot \sin\left(0.5 \sqrt{16 / \kappa_{a} - 1} \cdot \ln\left( \kappa_{a} \sigma_{ea} + \mbox{const} \right) \right)$. The same for the $y$ coordinate. The 3rd case is recognized as logarithmic spirals for certain choice of constants $C$.} and used in form \eqref{eq:2d_case_equiaffine_via_eq_aff_curvature} of the system \eqref{eq:2d_case_equiaffine} to show that logarithmic spirals found earlier for the equi-affine parametrization are the only solutions with constant affine curvature.

% Derived differential equations correspond to arbitrary degree of smoothness of the trajectory and arbitrary, not necessarily Euclidian, geometric invariance of the arc measurement along trajectory's path.
%

%
\par Geometric invariance and smoothness of contours are also relevant for the visual system. In this respect the primary visual cortex (V1) can be viewed as the bundle of what
are called 1-jets of curves in $\mathbb{R}$ \cite{Petitot:2003}. The 1-st order jet of a function $f$ is
characterized by three slots: the coordinate $x$, the value of $f$ at $x$, $y = f(x)$, and the
value of its derivative $p = f'(x)$. The latter is the slope of the tangent to the graph
of $f$ at the point $a = (x, f(x) )$ of $\mathbb{R}$. ``Jets are feature detectors specialized in the
detection of tangents. The fact that V1 can be viewed as a jet space explains why
V1 is functionally relevant for contour integration. ... The Frobenius integrability
condition ... is an idealized mathematical version of the Gestalt principle of good
continuation'' \cite{Petitot:2003}. Smooth drawings possess nice integrability properties. Edge completion as the interpolation of gaps between edge segments, which are extracted from an image, can
be performed by parabolas \cite{Handzel_Flash_2001_}. Smoothing may be applied by the motor system at the transitions between neighboring superimposed movement elements and thus the geometric levels of planning may precede the temporal level (see also \cite{Torres.Andersen:2006}).
%
%

% \par A prominent model in motor control -- the minimum-jerk \cite{Hogan:1984, Flash_Hogan_1985} -- corresponds to the 3rd order smoothness of the trajectory. However other degrees of smoothness were also considered in motor control studies: 2nd degree (minimum-acceleration principle) (eg. \cite{Ben_Itzkah.Karniel:2008}) and 4th degree (minimum snap principle) (eg. \cite{Richardson_Flash_2003}). Principles involving degrees of smoothness above 4 might be found relevant in future studies. Certain solutions of the equations are considered in this paper for dimensions $L=2,\, 3$, trajectory smoothness of different degrees, and parametrization in 6 geometries. Some solutions satisfy equations with all degrees of smoothness $n \geq 2$, like parabolas and circles (Table \ref{Table_mathematical_objects}). Solutions from the class of logarithmic spirals, however, are less ``universal'' as they are different for different degrees of smoothness and different parameterizations (Table \ref{Table_mathematical_objects}).
% Considering systems of equations with different combinations of dimension, order of trajectory smoothness and geometric parametrization could lead to more solutions whose candidacy for being geometric movement primitives can be further analyzed.

\par Different levels of smoothness at different movement's stages may be employed by the neural system even for well-practiced performance. In that case solutions of the system \eqref{eq:general_form_higher_order} with different degrees of smoothness $n$ might be combined together. Parabolic shapes satisfy equations \eqref{eq:general_form_higher_order} and constitute the class of the only affine invariant solutions for the case of the minimum-acceleration and minimum-jerk cost functionals \cite{Polyakov:2001, Polyakov:2006, Polyakov_et_al_B.Cyb:2009}. Possibility of non-parabolic equi-affine invariant solutions for higher degrees of smoothness has to be checked. Nevertheless, existence of non-parabolic solutions points to non-parabolic candidates for primitive shapes. Use of non-parabolic primitive shapes in production of complex movements might be efficient, for example, for movement segments which presumably are not represented solely in equi-affine geometry. Connection between the mechanisms of (1) action and (2) perception, and relevance of geometric invariance and smoothness for both mechanisms were observed in earlier studies. Therefore the proposed method of identifying geometric primitives may be meaningful for both motor and visual systems.
\par Obviously, minimum-jerk cost, as a scalar product, is invariant under Euclidian transformations. Left hand side of the upper equations in the systems \eqref{eq:general_form_lower_order}, \eqref{eq:general_form_higher_order} is composed of scalar products and therefore upper equations' solutions are invariant (remain solutions) under Euclidian transformations and uniform scaling with and without reflection but not general transformations that do not preserve orthogonality, e.g. arbitrary equi-center-affine transformation and transformations containing them (center-affine, equi-affine, affine). Therefore invariance of solutions was indicated in earlier works \cite{Polyakov:2006, Polyakov_et_al_B.Cyb:2009} and is indicated here in Table \ref{Table_mathematical_objects}. Minimum-jerk trajectories constrained with via-points are not invariant under arbitrary equi-affine transformations (and correspondingly affine) \cite{Polyakov:2001}. A recent follow up study \cite{Meirovitch:2014} thoroughly analyzed geometric invariance of the minimum-jerk trajectories with via-point.

\subsection*{Representation of movements in different geometries, compositionality, variability, and decision-making}
\par Straight point-to-point trajectory is described in the framework of the minimum-jerk model as 5th order polynomial of time \cite{Hogan:1984, Flash_Hogan_1985}:
\begin{equation}\label{eq:minimum_jerk_point_to_point}
    \bm{r}(t) = \bm{r}_{0} + \bm{A} \left[ 10 \left(\frac{t - t_0}{t_{d}}\right)^3 - 15 \left(\frac{t - t_0}{t_{d}}\right)^4 + 6 \left(\frac{t - t_0}{t_{d}}\right)^5\right]\, ,\;  t_0 \leq t \leq t_0  + t_d\, ,
\end{equation}
where $t_{d}$, $t_0$ and $\bm{A}$ are movement's duration, starting time, and the vector connecting movement's start ($\bm{r}_0$) and end points respectively. Trajectories described by \eqref{eq:minimum_jerk_point_to_point} have symmetric bell-shaped tangential velocity profile \cite{Flash_Hogan_1985}.
\par Minimum-jerk trajectories connecting two end-points and passing through a via-point \cite{Flash_Hogan_1985} are parabolic-like \cite{Polyakov:2006, Polyakov_et_al_B.Cyb:2009, Shpigelmacher:2006}. Such parabolic-like trajectories can be approximated based on vectorial composition of 3 point-to-point (straight) minimum-jerk movements, each having its own amplitude, duration, starting time and direction\footnote{This is similar to vectorial concatenation of two point-to-point movement elements into a single trajectory (without rest point in the middle) in the double-step paradigm \cite{Flash_Henis_1991}.} \cite{Polyakov:2006, Polyakov_et_al_PLoS_C_B:2009}. Composed triplets of point-to-point minimum-jerk trajectories were found to fit well piece-wise parabolic monkey scribbling movements \cite{Polyakov_et_al_B.Cyb:2009}. Moreover, a (straight) point-to-point minimum-jerk trajectory can be approximated with composition of 3 smaller and slower point-to-point minimum-jerk trajectories with each of 3 having the same direction, amplitude and duration \cite{Polyakov:2006, Polyakov_et_al_PLoS_C_B:2009}. So piece-wise parabolic trajectory can be decomposed in a hierarchical manner into short point-to-point minimum-jerk trajectories\footnote{Such decomposition is not just an exemplar case but this is a universal property for parabolic segments because it is preserved under arbitrary affine transformations and affine transformations can be used to map the piece of parabola approximating the path into an arbitrary parabolic segment.}, the algorithm is described further in text.
\par Straight segments form primitive shapes in Euclidian geometry. For other geometries considered in this work straight segments have either zero arc (equi-affine, equi-center-affine\footnote{Consider lines translated to cross the origin.}, similarity arcs) or their arc is not defined at all (center-affine and affine arcs). Correspondingly parabolic segments are primitive shapes in equi-affine geometry, in other geometries considered in this work parabolic segments either have zero arc (affine arc) or do not satisfy equations \eqref{eq:general_form_lower_order}, \eqref{eq:general_form_higher_order}. So construction of parabolic-like (equi-affine) primitives can be based on sequential representation of Euclidian primitives getting coarticulated! Such construction demonstrates a plausible realization of simultaneous representation of movements in different geometries. Representation of movements in several geometries suggested by Bennequin et. al. \cite{Bennequin:2009, Fuchs:2010} as weighted mixture of different geometric arcs and proposed as a basis for geometric rationale of movement timing can also be viewed through the (coexisting) perspective of coarticulation of primitives in one geometry into primitives in another geometry. Indeed, analysis of monkey scribblings showed that jerky unordered movements converged into organized piece-wise parabolic performance \cite{Polyakov:2001, Polyakov_etc_2001, Polyakov:2006, Polyakov_et_al_B.Cyb:2009, Polyakov_et_al_PLoS_C_B:2009}, while sequences of point-to-point trajectories by humans got coarticulated into sequences of smooth trajectories \cite{Sosnik_etc_2004} that are parabolic-like \cite{Polyakov:2006, Shpigelmacher:2006,Polyakov_et_al_B.Cyb:2009}.
\par Variability of well-practiced spontaneous monkey scribbling movements was influenced by getting or not getting a reward \cite{Polyakov:2006, Polyakov_et_al_PLoS_C_B:2009}. Tuning of  primitives' onset in different kinds of goal-directed movements (achieving a prescribed movement goal or implementing spontaneous search for invisible target) may be guided by decision-making and/or action selection based on ongoing feedback/reinforcement signals, e.g. receiving or not receiving a reward. Therefore greater variability of non-rewarded movements was interpreted as characterizing monkey's decision-making about concatenating concurrent, already preplanned, piece-wise parabolic movement sequence with another primitive element of monkey's scribbling repertoire. It was suggested that paradigms involving decision-making might be advantageous in studies investigating movement construction based on the compositionality of movement primitives \cite{Polyakov:2006, Polyakov_et_al_PLoS_C_B:2009}. Below algorithm describes how trajectories composed of (1) parabolas and (2) straight lines, both classes are invariant under affine transformations, may be composed of identical elementary building blocks. Different syntactic rules of combining higher level primitives (e.g. piece-wise parabolic sequences) may be developed during practice.
\subsection*{Hierarchical construction of complex trajectories based on primitives in Euclidian and equi-affine geometries}
\par Assume that a rule of implementing a slow and short point-to-point minimum-jerk movement of the end-effector could be built in, eg. in the nervous system of humans or for a robotic arm. Let the peak speed of such rule of motion be equal to $v_p$. Proceed as follows\footnote{Parts of the method were reported and illustrated in \cite{Polyakov:2006, Polyakov_et_al_PLoS_C_B:2009}.}.
\begin{enumerate}
  \item Segment the trajectory into approximately straight and curved portions.
   \item\label{item2_hierarchical construction} Approximate curved portions of the trajectory path with parabolic segments. The less is their curvature at the ``vertex'' the ``wider'' is parabola, while ``width'' is measured with parabola's focal parameter.
   \item Decompose each parabolic portion of the trajectory into three point-to-point trajectories.
   \item Decompose identified straight portions of the trajectory into short and slow point-to-point trajectories: for each straight portion $i$ compute peak speed $v_{p,\, i}(\bm{A}, t_d)$ given its $\bm{A}_i$ and $t_{d,\, i}$ from \eqref{eq:minimum_jerk_point_to_point}; the number of lower hierarchical levels for each portion will be approximately equal to $\ln(v_p / v_{p,\, i}) / \ln(0.55)$. The slow and short rule of motion (assumed to be known) lies at the lowest level of hierarchy of each decomposed straight movement segment.
\end{enumerate}

\par Parabolic-like trajectories and relatively long straight movements form the top of the hierarchical pyramid, each level below the top consists of shorter and slower straight trajectories; the number of elements is 3 times the number of elements at the level above. The peak speed of the point-to-point trajectories $\approx$ 0.55 of the peak speed of the longer point-to-point trajectories one level above. One may attempt to perturb the decomposition during a number of trials in order to construct more accurate and optimal performance.
\par Given that a point-to-point trajectory \eqref{eq:minimum_jerk_point_to_point} connects two states of rest by definition one may allow certain parts of the longer straight movements to overlap, eg. an original straight portion of the trajectory and a point-to-point trajectory obtained during decomposition of the parabolic-like segment. By means of such overlaps it is possible to avoid intermediate states of rest during the entire piece-wise parabolic trajectory.

\subsection*{What happens when the motor control system is compromised}
\par View on motor output when the system is compromised may provide an additional insight. Motor control studies of patients suffering from Parkinson's decease (PD) observed and quantified violations of known motor regularities like the two-thirds power-law, isochrony, kinematic smoothness. Compliance with the isochrony principle \cite{viviani.terzuolo:1982} was impaired for the PD patients versus the control group in experiment involving point-to-point movement via an intermediate target \cite{Flash.Henis.Inzelberg.Korczyn:1992}. The same study also reported that patients' velocity profiles demonstrated substantial abnormalities including lack of smoothness and multiple small peaks or plateaus in the velocity profile.
\par In the framework of geometric invariance the two-thirds power-law is equivalent to constant rate of accumulating equi-affine arc thus implying movement timing being proportional to accumulated equi-affine arc. Apparently, patients with PD demonstrated impairments in how the two-thirds power-law characterizes their perception of planar motion \cite{Dayan.Inzelberg.Flash:2012}. In particular, patients with PD perceived on average movements closer but not equal to a constant Euclidian velocity as more uniform than movements with constant equi-affine velocity, in contrast to choices of control subjects \cite{Dayan.Inzelberg.Flash:2012}. In compliance with other studies mentioned in \cite{Dayan.Inzelberg.Flash:2012} this result demonstrates central, eg. visuo-motor, and not purely motor impairment of PD patients and supports again central (not purely motor) role of geometric characteristics of biological motion. Supporting the central role of geometric invariance, fMRI study of healthy humans demonstrated that basal ganglia respond preferentially to visual motion with constant rate of accumulating the equi-affine arc \cite{Dayan.etal.Flash:2007} while basal ganglia is also the main location of disfunction in PD. In terms of the mixed-geometry approach by Bennequin et. al. \cite{Bennequin:2009} one can hypothesize that dominance of the equi-affine contribution to movement representation is replaced with more dominant contribution of Euclidian measurement of trajectory's arc in case of PD patients.
\par Task-incidental degrees of freedom values of PD patients (while off dopaminergic medication) were abnormally variable during automated movements while task-relevant components abnormally dominated patients' intentional motions \cite{Torres.Heilman.Poizner:2011}. Moreover, patients' transition between voluntary and automated modes of joint control was abrupt, and, unlike normal controls, the type of visual guidance differentially affected their postural trajectories. These finding provided support to the view that PD patients lack automated control that contributes to impairments in voluntary control and that basal ganglia are critical for multi-joint control \cite{Torres.Heilman.Poizner:2011}. A different study demonstrated that for PD patients attention induced a shift from the automatic mode to the controlled pattern within the striatum\footnote{A component of the basal ganglia.} while for the control subjects attention had no apparent effect on the striatum when movement achieved the automatic stage \cite{Wu.Liu.Zhang.Hallett.Zheng.Chan:2014}. Basal ganglia contributes to decision making processes including decisions related to perception and action (eg. see \cite{Berns.Sejnowski:1996, Cheng.Anderson:2012, Ding.Gold:2013}).
\par Given a plausibly intimate relationship between movement variability and decision-making in the framework of movement compositionality (no need for decision-making during completing a preplanned well-practiced motor program would reduce variability), to my view, non-typical variability patterns demonstrate that PD may impair the decision-making process for choosing primitives composing movements. Disfunction of this decision-making process may also contribute to such PD disorders as lack of smoothness and irregularity of the velocity profiles leading in turn to (at least partial) failure to exploit more advantageous parsimonious representation provided by  invariance of geometric primitives\footnote{In the framework of approach of this manuscript geometric primitives provide more parsimonious representation for smooth movements}. Another cause to observed movement variability features of PD patients could be disruption of their ability to coarticulate sequences of elementary submovements into complex and smooth task dependent primitives. Apparently, study with the double-step paradigm reported that the PD patients have impaired abilities to process simultaneously the motor responses to two visual stimuli which are presented in rapid succession \cite{Plotnik.Flash.etc:1998}. I hypothesize that during practice intact motor control system tends to achieve motor performance based on smooth and stereotypical (often following coarticulation) patterns characterized with convergence to low variability and low-dimensional representation that manifests the principle of greater parsimony \footnote{The motor control system tends to achieve more parsimonious control strategies through practice/learning} \cite{Polyakov:2006}.
\par Assessment of behaviors by means of measuring stochastic properties of intra-trial variability based on special characteristic (micro-movements)  was proposed recently in the framework of autism spectrum disorders (ASD) and PD studies \cite{Torres.etal:2013, Torres:2013B}. The method successfully characterized behaviors of subjects.  Apparently, ASD results, in part, from impaired basal ganglia function (eg. see \cite{Qiu.Adler.etc:2010, Prat.Stocco:2012}). Probably, mechanisms related to geometric invariance and decision making could also be impaired to certain degree in case of ASD. Knowing the differences between judgements of observed movement speed uniformity between the PD patients (whose basal ganglia function is impaired) and control subjects \cite{Dayan.Inzelberg.Flash:2012}, it would be interesting to implement similar experiment with ASD patients and to assess their location on the scale Euclidian -- equi-affine uniformity versus control and PD subjects. Another interesting approach could be to implement setup of compositionality and movement primitives studies with ASD and PD patients, for example setups from the studies of movement coarticulation \cite{Sosnik_etc_2004} and point of no return \cite{Polyakov:2006, Sosnik.Shemesh.Abeles:2007}, and to compare the performance of the 3 types of subjects: PD, ASD and controls.
\par Impairment of decision-making mechanisms employed in motor control may disrupt a plausible hierarchical procedure (described above) of constructing smooth complex trajectories from geometric primitives. Such impairment may also destroy higher level mechanisms of (1) binding primitives from different geometries and  (2) Bennequin's time representation based on weighted mixture of different geometric arcs.  Lack of smoothness is a typical consequence of a vast range of neurological disorders, for example stroke, ataxia, Huntington's disease, secondary parkinsonism. Probably, some of those disorders are characterized with certain levels of disfunction in movement compositionality.
\par I am not aware about studies analyzing equi-affine and affine invariants of movement trajectories produced by the patients with neurodegeneration  (eg. PD) or neurodevelopmental disorder (eg. ASD). Numerical computations of such quantities \cite{Calabi_Olver_Tannenbaum_1996} are highly sensitive to non-smoothness and irregularities in trajectories' data, require data regularization\footnote{In particular because of high order differentiation. Text S1 of \cite{Polyakov_et_al_PLoS_C_B:2009} describes regularization procedure used in equi-affine analysis of monkey scribbling movements.} and therefore would be especially challenging for motor output produced by patients characterized by non-smooth movements. Still, in the current work, geometric invariance of movement primitives is associated with simultaneous ability to successfully coarticulate basic movement elements into smooth movement blocks after practicing novel motor task.

\subsection*{Afterword}
\par To my view the following insight of a prominent mathematician of the 20-th century Andrey Kolmogorov anticipated the idea of geometric movement primitives \cite{Kolmogorov:1988}: ``If we turn to the human activity --  conscious, but not following the rules of formal logic, i.e. intuitive or semi-intuitive activity, for example to motor reactions, we will find out that high perfection and sharpness of the mechanism of continuous motion is based on the
movements of the continuous-geometric type ... One can consider, however, that this is
not a radical objection against discrete mechanisms. Most likely the intuition of continuous
curves in the brain is realized based on the discrete mechanism''\footnote{Translated from Russian by FP.}.
\par The way of representing the ``continuous curves in the brain'' as coarticulated geometric primitives might go beyond planning trajectory paths and may correspond to perception processes and geometric imagination as well. Moreover, I speculate that at certain hierarchical level of cognitive processes the ``discrete mechanisms'' of complex movements and language intersect. Observations of low-dimensional representation of monkey scribbling movements with parabolic primitives and reward-related concatenation of parabolic segments into complex trajectories \cite{Polyakov:2006, Polyakov_et_al_PLoS_C_B:2009} support feasibility of this speculation.

\counterwithin{equation}{section}

\appendix
\setcounter{table}{0}
\renewcommand\thetable{\thesection.\arabic{table}}

\setcounter{example}{0}
\renewcommand{\theexample}{\Alph{section}.\arabic{example}}
\setcounter{prop}{0}
\renewcommand{\theprop}{\Alph{section}.\arabic{prop}}

\section{Derivation of the proposition}\label{Appendix_proof_propostion}

\par Here upper equation from the system \eqref{eq:general_form_lower_order} is derived for the curves belonging to the classes $\mathcal{\tilde{A}}_{n,\, L}$ and $\mathcal{A}_{n,\, L}$. The derivation of the equations proofs Proposition \ref{main_proposition}.

\subsubsection*{Derivation}
\par Geometric parametrization $\sigma$ of a curve is given. The rule $\sigma(t)$ of accumulating $\sigma$ with time along the curve is strictly monotonous and differentiable as many times as necessary. Noting that there is one-to-one continuous correspondence between $t$ and $\sigma$,
for the function $\sigma(t) \in [0,\, \Sigma]$ define an
inverse function $t = \tau(\sigma) \in [0,\, T]$. The following notation is used:
$$v \equiv v(\sigma) \equiv  \left. \frac{d}{dt}\sigma(t) \right|_{t = \tau(\sigma)}  \equiv \left. \dot{\sigma}(t)\right|_{t = \tau(\sigma)}\, .$$
Further the following property based on the chain rule is used for a differentiable function $f$:
\begin{equation}\label{eq:general_formula_df_dt}
    \frac{d}{dt}f(\sigma(t)) = \frac{d\sigma}{dt} \cdot \frac{d}{d\sigma}f(\sigma) \equiv \dot{\sigma}\frac{d}{d\sigma}f(\sigma) \equiv v f'\, ,
\end{equation}
where prime denotes differentiation with respect to $\sigma$.
So, for example, two higher order derivatives of $\sigma$ with respect to time will be:
\begin{eqnarray}\label{example_w_j}
 w &=& w(\sigma) \equiv \left. \frac{d^2}{dt^2}\sigma(t) \right|_{t = \tau(\sigma)}= \left. \frac{d}{dt}\left[v(\sigma(t))\right]\right|_{t = \tau(\sigma)} =  v \frac{d}{d\sigma}v = {v}'v \\
 \nonumber j &=& j(\sigma) = \left. \frac{d^3}{dt^3}\sigma(t) \right|_{t = \tau(\sigma)}= v\frac{d}{d\sigma}w =
     {v}''{v}^2 + {{v}'}^2 v \, .
\end{eqnarray}
\par Without limitation of generality further derivations will be implemented for trajectories in 2 dimensions. Derivations for the trajectories in 3 or higher dimensions are identical. So consider $J_{\sigma}(\bm{r}_L,\, n)$ from \eqref{general_cost_function1} with $L = 2$ and use \eqref{eq:general_formula_df_dt} to implement change of variables:
\begin{eqnarray}\label{general_cost_function}
    J_{\sigma}(\bm{r}_2,\, n) & = &\displaystyle\int\limits_0^T \left\{
                                \left[\frac{d^n x(\sigma(t))}{dt^n}\right]^2 + \left[\frac{d^n y(\sigma(t))}{dt^n}\right]^2\right\}\, dt = \\
                              & &  \int\limits_0^{\Sigma}\frac{1}{v} \left.\left\{
                                \left[\frac{d^n x(\sigma(t))}{dt^n}\right]^2 + \left[\frac{d^n y(\sigma(t))}{dt^n}\right]^2\right\}\right|_{t = \tau(\sigma)} d\sigma = \nonumber\\
                         & &   \int\limits_0^{\Sigma}\frac{1}{v} I_{n}(x', x'', \ldots, x^{(n)}; y', y'', \ldots, y^{(n)}; v, v', \ldots, v^{(n-1)}) d\sigma \;, \nonumber
\end{eqnarray}
where  $I_{n}$ denotes the expression parameterized with $\sigma$:
\begin{equation}\label{eq:In}
    I_{n} \equiv \left. \left[\frac{d^n x(\sigma(t))}{dt^n}\right]^2 + \left[\frac{d^n y(\sigma(t))}{dt^n}\right]^2 \right|_{t = \tau(\sigma)}\; .
\end{equation}
For example:
\begin{example} In case $n=3$, one has:
\begin{eqnarray*}
J_{\sigma}(\bm{r}_2, 3) & = & \int\limits_0^T ({\dddot{
                                          x}}^2 + {\dddot{y}}^2 ) dt =\int\limits_0^{\Sigma}\frac{1}{v}\bigl[({x'''}^2+{y'''}^2)
      {v}^6 + \\
 & & 9({x''}^2+{y''}^2){w}^2{v}^2 +({x'}^2+{y'}^2){j}^2 + 6(x'''x''+y'''y''){v}^4w +\\
 & & 2(x'''x'+y'''y'){v}^3j + 6(x''x'+y''y') v w j \bigr] \, d\sigma = \\
     & & \int\limits_0^{\Sigma} \frac{1}{v} \cdot I_{3}(x', x'', x'''; y', y'', y'''; v, v', v'')\,
     d\sigma \,
\end{eqnarray*}
with $w$, $j$ from \eqref{example_w_j}.  \;\;\; $\Box$
\end{example}
\par I approach the optimization problems \eqref{solution_path_minimization_problem2} and \eqref{solution_path_minimization_problem1} with a standard method from the calculus of variations, the Euler-Poisson (E-P) equation with Lagrange multiplier (eg. \cite{Gelfand_Fomin_2000}). The Lagrange multiplier ($\lambda$) is used to guarantee that the speed of accumulating the arc is feasible: $\displaystyle\int_{0}^{\Sigma} \frac{d\sigma}{v} = T$.
\begin{eqnarray}\label{E_P_equation_nth_order}
     \mbox{E-P}(I_n / v) &=& \frac{\partial (I_n / v)}{\partial v} -
                    \frac{d}{d\sigma}\left( \frac{(\partial (I_n / v))}{\partial v'}
                    \right)+ \frac{d^2}{d\sigma^2}\left( \frac{\partial (I_n / v)}{\partial v''}\right)
                    - \ldots  \\
    \nonumber    & & + (-1)^{n-1} \frac{d^{n-1}}{d\sigma^{n-1}} \left(\frac{\partial (I_n / v)}{\partial v^{(n-1)}} \right) + \lambda\frac{\partial}{\partial v} \left(\frac{1}{v} \right)= \\
     \nonumber &=& v^{(2n - 3)} (\ldots) + v^{(2n - 4)}(\ldots) + \ldots + v' (\ldots) \\
    \label{E_P_equation_nth_order_mu}  & &  +  v^{2n - 2}\left(\ldots\right)+ \lambda\frac{\partial}{\partial v} \left(\frac{1}{v} \right) = 0 \, . % \left(x^{(n)}\right)^2 + \left(y^{(n)}\right)^2 +
\end{eqnarray}
Expressions in brackets $(2 n - 3)$, $(2 n - 4)$, etc. from \eqref{E_P_equation_nth_order_mu} denote differentiation with respect to $\sigma$ of corresponding order. All derivatives of $v$ in the brackets $(\ldots)$ in \eqref{E_P_equation_nth_order_mu} have order lower than the order of derivative of $v$ multiplying the brackets. Note that the term $v^{2n - 2}$ in \eqref{E_P_equation_nth_order_mu} represents the value of the speed to the power of $2n - 2$ and not the order of derivative. So the expression
$$ v^{2n - 2}\left( \ldots\right) + \lambda\frac{\partial}{\partial v} \left(\frac{1}{v} \right)$$ % \left(x^{(n)}\right)^2 + \left(y^{(n)}\right)^2 +
is the only part of \eqref{E_P_equation_nth_order_mu} which contains no derivatives of $v$. Let us denote by $\mu_n$ the expression in the brackets multiplied by $ v^{2n - 2}$, so the Euler-Poisson equation \eqref{E_P_equation_nth_order_mu} can be rewritten as follows:
\begin{equation}\label{mu_order_n}
    \mbox{E-P}(I) = v^{(2n - 3)} (\ldots) + v^{(2n - 4)}(\ldots) + \ldots + v' (\ldots) +  v^{2n - 2} \mu_n + \lambda\frac{\partial}{\partial v} \left(\frac{1}{v} \right) = 0
\end{equation}
Now as an
\begin{example} Consider the case of minimum-jerk criterion, in other words the 3rd order smoothness ($n = 3$). Using derivations identical to the ones in \cite{Polyakov:2006, Polyakov_et_al_B.Cyb:2009}, the Euler-Poisson equation corresponding to \eqref{E_P_equation_nth_order_mu} will be as follows:
\begin{eqnarray*}
     & & v'''\cdot (\ldots) + v'' \cdot (\ldots) + v'
     \cdot(\ldots) + \\
     & & + v ^4 \cdot({x'''}^2 + {y'''}^2 - 2x''x^{(4)} - 2y''y^{(4)} + 2x'x^{(5)} + 2y'y^{(5)}) + \lambda\frac{\partial}{\partial v} \left(\frac{1}{v} \right) = 0 \, .
\end{eqnarray*}
Here
$$\mu_3 = {x'''}^2 + {y'''}^2 - 2x''x^{(4)} - 2y''y^{(4)} + 2x'x^{(5)} + 2y'y^{(5)} \, .\;\;\; \Box$$
\end{example}
\par The desirable $v$ for the optimal solution is constant,
according to (\ref{eq_goal_set2}). Therefore all derivatives of $v$ are zero and the Euler-Poisson equation for the desired $v$ reduces to the following:
$$ v^{2n - 2} \mu_n - \frac{\lambda}{v^2} = 0 \;.$$
As stated above, $v$, $\lambda$ are constant, therefore under the assumption $v \neq 0$ which obviously takes place,
\begin{equation}\label{EP_equation_reduced_form}
   \mu_n  = \mbox{const} \, .
\end{equation}
\begin{prop}\label{p:mu_n}
\begin{eqnarray*}
  \nonumber \mu_n & = &  \left[x^{(n)}\right]^2 + \left[y^{(n)}\right]^2 - 2 \left[ x^{(n-1)} x^{(n+1)} + y^{(n-1)} y^{(n+1)} \right] + 2 \left[ x^{(n-2)} x^{(n+2)} + y^{(n-2)} y^{(n+2)} \right]\\
                  &  + & \ldots
         + (-1)^{n-1} \cdot 2 \left[ x' x^{(2n - 1)} + y' y^{(2n - 1)} \right]\, ,
\end{eqnarray*}
or more formally
\begin{equation}\label{mu_full_expression}
  \mu_n = \left[x^{(n)}\right]^2 + \left[y^{(n)}\right]^2 + 2 \sum_{i = 1}^{n-1} (-1)^i \left(x^{(n - i)} x^{(n + i)} + y^{(n - i)} y^{(n + i)}\right)\; ,
\end{equation}
which is the 2-dimensional version of the upper equation in system \eqref{eq:general_form_lower_order}.
\end{prop}
\begin{proof}
\par In order to find the expression for $\mu_n$ implement the differentiation in the Euler-Poisson equation \eqref{E_P_equation_nth_order}. Apparently, the argument of the cost functional from \eqref{general_cost_function},
$$ \left[\frac{d^n x(\sigma(t))}{dt^n}\right]^2 + \left[\frac{d^n y(\sigma(t))}{dt^n}\right]^2 $$
can be split into the $''x``$ and $''y``$ parts and therefore the functional argument $I_n$ of the Euler-Poisson equation is splittable as well:
\begin{equation}\label{eq:splitting_I_Ix_Iy}
I_n = I_{n,\, x} + I_{n,\, y} = \left.\left\{\left[\frac{d^n x(\sigma(t))}{dt^n}\right]^2 + \left[\frac{d^n y(\sigma(t))}{dt^n}\right]^2\right\}\right|_{t = \tau(\sigma)} \, .
\end{equation}
The result of differentiation in the $x$ part
\begin{equation}\label{eq:x_part}
    I_{n,\,x} = \left. \left[\frac{d^n x(\sigma(t))}{dt^n}\right]^2 \right|_{t = \tau(\sigma)}
\end{equation}
is identical to the result of differentiation in the $y$ part $I_{n,\,y} = \left. \left[\frac{d^n y(\sigma(t))}{dt^n}\right]^2 \right|_{t = \tau(\sigma)}$ up to the name of the argument ($x$ being replaced by $y$).
\par So without limitation of generality I implement the proof for the $''x``$ part only and need to prove that
\begin{eqnarray}\label{eq:EP_x_part}
(E-P) \frac{I_{n,\, x}}{v} &= & (E-P)\left\{\frac{1}{v} \left.\left[\frac{d^n x(\sigma(t))}{dt^n}\right]^2\right|_{t = \tau(\sigma)} \right\} =  \\
\nonumber & = & v^{2n - 2}\left(\left(x^{(n)}\right)^2 - 2  x^{(n-1)} x^{(n+1)} + 2 x^{(n-2)} x^{(n+2)}+ \ldots + (-1)^{n-1} \cdot 2 x' x^{(2n - 1)}\right)\, \\
\nonumber & + & v' (\ldots) + v''(\ldots) + \ldots + v^{(n)} (\ldots) + \lambda_x\frac{\partial}{\partial v} \left(\frac{1}{v} \right)\, .
\end{eqnarray}
The result for the $y$ part being identical to \eqref{eq:EP_x_part} with proper replacement of $x$ terms with $y$ terms described above will immediately imply equality \eqref{mu_full_expression} which I am proving.
\par Now the expression for $d^nx(\sigma(t)) / dt^n$ will be rewritten and parameterized with $\sigma$. Time derivatives of $x(\sigma(t))$ parameterized by $\sigma$ are computed as follows:
\begin{eqnarray*} % \label{eq:dx_dt}
    \left.\dot{x}\right|_{t = \tau(\sigma)} & = & \left.\frac{dx}{d\sigma} \frac{d\sigma}{dt}\right|_{t = \tau(\sigma)} = x' v \\
    \left.\ddot{x}\right|_{t = \tau(\sigma)} & = & v \cdot \left(\left.\dot{x}\right|_{t = \tau(\sigma)}\right)' = x'' v^2 + x' v' v \\
    \left.\dddot{x}\right|_{t = \tau(\sigma)} & = & v \cdot \left(\left.\ddot{x}\right|_{t = \tau(\sigma)}\right)' = x''' v^3 + 3 x'' v' v^2 + x' v'' v^2 + {v'}^2 (x' v)\, .
\end{eqnarray*}

It can be shown by induction that
\begin{equation}\label{eq:der_x_dt_order_n}
    \left.\frac{d^n x(\sigma(t))}{dt^n}\right|_{t = \tau(\sigma)} = v^{n - 1} \left[ \sum_{k = 2}^{n} \binom{n}{k}x^{(n - (k - 1))} v^{(k - 1)} + x^{(n)} v\right] + \sum_{i,\, j > 0} v^{(i)} v^{(j)} (\ldots)\; .
\end{equation}
Expressions denoted by $(\ldots)$ and multiplied by the product of derivatives of $v$, $v^{(i)} v^{(j)}$, in \eqref{eq:der_x_dt_order_n} are irrelevant in our derivations as their contribution to $I_{n,\, x}$ will be zeroed out under the assumption of constant speed $(v = \mbox{const})$.
\par Expression \eqref{eq:der_x_dt_order_n} implies for the squared derivative:
\begin{eqnarray}\label{eq:der_x_dt_order_n_squared}
    \left.\left[\frac{d^{n} x(\sigma(t))}{dt^{k}}\right]^2\right|_{t = \tau(\sigma)} & = &  v^{2n - 2}\left[ \sum_{k = 2}^{n} \binom{n}{k}x^{(n - (k - 1))} v^{(k - 1)} + x^{(n)} v\right]^2 \\
    \nonumber &+& \sum_{i,\, j > 0} v^{(i)} v^{(j)} (\ldots) = \sum_{i,\, j > 0} v^{(i)} v^{(j)} (\ldots) \\
    \nonumber & +& v^{2n - 2}\left[ \left( x^{(n)}\right)^2 v^2 + 2 v x^{(n)} \sum_{k = 2}^{n} \binom{n}{k}x^{(n - (k - 1))} v^{(k - 1)} \right]  \; .
\end{eqnarray}
So for $I_{n,\, x}$ from \eqref{eq:x_part}
\begin{eqnarray}\label{eq:I_n_x_expressed_x_v}
    \frac{I_{n,\, x}}{v} =  v^{2n - 1} \cdot \left[x^{(n)}\right]^2 + 2 \cdot v^{2n - 2} \cdot x^{(n)}  \sum_{k = 2}^{n} \binom{n}{k} x^{(n - (k - 1))} v^{(k - 1)} + \sum_{i, j > 0} v^{(i)} v^{(j)} (\ldots) \;\;\;\;\;\;\;
\end{eqnarray}
and the Euler-Poisson equation \eqref{E_P_equation_nth_order_mu} for $I_{n,\, x} / v$ will be as follows:
\begin{eqnarray}
    \nonumber (E-P) \frac{I_{n,\, x}}{v} & = & \lambda_x \frac{\partial}{\partial v} \left( \frac{1}{v} \right) +    (2n - 1) v^{2n - 2} \left[x^{(n)} \right]^2 - 2 \binom{n}{2} \cdot \frac{d}{d\sigma}\left[v^{2n - 2} \cdot x^{(n)} \cdot x^{(n - 1)} \right] \\
    \nonumber & + &
     2 \binom{n}{3} \cdot \frac{d^2}{d\sigma^2}\left[v^{2n - 2} \cdot x^{(n)} \cdot x^{(n - 2)} \right] +
     \ldots \\
     \nonumber & + & (-1) ^{n-1} \cdot 2 \cdot  \binom{n}{n} \cdot \frac{d^{n - 1}}{d\sigma^{n - 1}}\left[v^{2n - 2} \cdot x^{(n)} \cdot x'\right]
     + \sum_{i > 0} v^{(i)} (\ldots)  =
\end{eqnarray}
\begin{eqnarray}
    \nonumber & &  \sum_{i > 0} v^{(i)} (\ldots) + \lambda_x \frac{\partial}{\partial v} \left( \frac{1}{v} \right) + (2n - 1) v^{2n - 2} \left[x^{(n)} \right]^2 + \\
    \nonumber & & 2 \sum_{k = 2}^{n}(-1)^{k - 1} \binom{n}{k}\cdot \frac{d^{k - 1}}{d\sigma^{k - 1}}\left[v^{2n - 2} \cdot x^{(n)} \cdot x^{(n-(k - 1))} \right]  = \\
    \nonumber & &  \sum_{i > 0} v^{(i)} (\ldots) + \lambda_x \frac{\partial}{\partial v} \left( \frac{1}{v} \right) +  (2n - 1) v^{2n - 2} \left[x^{(n)} \right]^2 + \\
     \label{eq:I_n_x_expressed_x_v} & & 2 v^{2n - 2} \sum_{k = 2}^{n}(-1)^{k - 1} \binom{n}{k}\cdot \frac{d^{k - 1}}{d\sigma^{k - 1}}\left[x^{(n)} \cdot x^{(n-(k - 1))} \right]  \, .
\end{eqnarray}
The values of binomial coefficients in \eqref{eq:I_n_x_expressed_x_v} form a slice of Pascal triangle without two numbers at the boundary. A property of Pascal triangle introduced in Proposition \ref{proposition_eq:expression_to_be_proven} of Appendix \ref{Appendix_Pascal_triangle} implies that
\begin{eqnarray*}
 (2n - 1) v^{2n - 2} \left[x^{(n)} \right]^2 + 2 v^{2n - 2} \sum_{k = 2}^{n}(-1)^{k - 1} \binom{n}{k}\cdot \frac{d^{k - 1}}{d\sigma^{k - 1}}\left[x^{(n)} \cdot x^{(n-(k - 1))} \right] = \\
     v^{2n - 2}\left(\left(x^{(n)}\right)^2 - 2  x^{(n-1)} x^{(n+1)} + 2 x^{(n-2)} x^{(n+2)}+ \ldots + (-1)^{n-1} \cdot 2 x' x^{(2n - 1)}\right)
\end{eqnarray*}
and so
\begin{eqnarray*}
    (E-P) \frac{I_{n,\, x}}{v} & = & \sum_{i > 0} v^{(i)} (\ldots) + \lambda_x \frac{\partial}{\partial v} \left( \frac{1}{v} \right) + v^{2n - 2} \left\{\left[x^{(n)}\right]^2 + 2 \sum_{i = 1}^{n - 1} (-1)^{i}  x^{(n - i)} x^{(n + i)} \right\} \; ,
\end{eqnarray*}
which completes the proof of Proposition \ref{p:mu_n} meaning that the Proposition \ref{main_proposition} is true.
\end{proof} 

\setcounter{table}{0}

\setcounter{example}{0}
\renewcommand{\theexample}{\Alph{section}.\arabic{example}}
\setcounter{prop}{0}
\renewcommand{\theprop}{\Alph{section}.\arabic{prop}}

\section{A property of Pascal's triangle}\label{Appendix_Pascal_triangle}
\par Binomial coefficients of the form $\binom{N}{k}$, $k = 0,\, \ldots,\, N$ form the $N$th row of Pascal's triangle. Denote the elements of the $N$th row of the triangle as $$\alpha_{N,\, k} \equiv \binom{N}{k} \, .$$
An example of the first 9 rows of Pascal's triangle is demonstrated in Table \ref{table:example_Pascal_triangle}. The first 5 rows of Pascal's triangle with coefficients $\alpha$ replacing the numbers are demonstrated in Table \ref{table:example_Pascal_triangle_alphas}.
\begin{table}[h]
% The example is based on the example from http://forum.mackichan.com/node/274
\begin{tabular}{rccccccccccccccccc}
$N=0$: & & & & & & & & & 1 & & & & \\
$N=1$: & & & & & & & & 1 & & 1 & & & \\
$N=2$: & & & & & & & 1 & & 2 & & 1 & & \\
$N=3$: & & & & & & 1 & & 3 & & 3 & & 1 & \\
$N=4$: & & & & & 1 & & 4 & & 6 & & 4 & & 1 & \\
$N=5$: & & & & 1 & & 5 & & 10 & & 10 & & 5 & & 1 &\\
$N=6$: & & & 1 & & 6 & & 15 & & 20 & & 15 & & 6 & & 1 & \\
$N=7$: & & 1 & & 7 & & 21 & & 35 & & 35 & & 21 & & 7 & & 1 & \\
$N=8$: & 1 & & 8 & & 28 & & 56 & & 70 & & 56 & & 28 & & 8 & & 1 \\
\end{tabular}
\caption{An example of the first 9 rows of Pascal's triangle.}\label{table:example_Pascal_triangle}
\end{table}
\begin{table}[h]
% The example is based on the example from http://forum.mackichan.com/node/274
\begin{tabular}{rccccccccc}
$N=0$: & & & & & $\alpha_{0,\, 0}$ & & & & \\
$N=1$: & & & & $\alpha_{1,\, 0}$ & & $\alpha_{1,\, 1}$ & & & \\
$N=2$: & & & $\alpha_{2,\, 0}$ & & $\alpha_{2,\, 1}$ & & $\alpha_{2,\, 2}$ & & \\
$N=3$: & & $\alpha_{3,\, 0}$ & & $\alpha_{3,\, 1}$ & & $\alpha_{3,\, 2}$ & & $\alpha_{3,\, 3}$ & \\
$N=4$: & $\alpha_{4,\, 0}$ & & $\alpha_{4,\, 1}$ & & $\alpha_{4,\, 2}$ & & $\alpha_{4,\, 3}$ & & $\alpha_{4,\, 4}$%
\end{tabular}
\caption{An example of the first 5 rows of Pascal's triangle when coefficients $\alpha$ replace the numbers.}\label{table:example_Pascal_triangle_alphas}
\end{table}
\par The values of $\alpha$ at the next level $i$ of the triangle are obtained recursively from the values at the level $i - 1$ according to the following rule:
\begin{eqnarray}
\label{eq:define_alpha1}
     \alpha_{i,\, j} = \alpha_{i - 1,\, j - 1} + \alpha_{i - 1,\, j}\; .
%      \\ \label{eq:define_alpha2}
%     \alpha_{i,\, j} = \sum_{k = 0}^{j}\alpha_{i - j - 1 + k,\, k}
\end{eqnarray}
while
\begin{equation}\label{eq:boundary_conditions_alpha}
        \alpha_{0,\, 0} = 1;\; \alpha_{-1,\, i} = 0;\; \alpha_{i,\, i + 1} = 0\; \forall i\, .
\end{equation}
I \textit{prove} a formula which establishes a relationship between the elements in a row of Pascal's triangle and in the ``diagonal'' whose left most element appears adjacent to the left most element of the sequence of coefficients in the row; an example is provided below.
\begin{prop}\label{proposition_eq:expression_to_be_proven}
\begin{equation}\label{eq:expression_to_be_proven}
    \sum_{k = i}^{N} (-1)^{i - k + 1} \alpha_{k - 1,\, i - 1} \cdot \alpha_{N,\, k}  = -1,\; i \leq N \,.
\end{equation}
\end{prop}
The numbers $\alpha_{k - 1,\, i - 1}$ belong to the ``diagonal'' and elements $\alpha_{N,\, k}$ belong to the row $N$. An example of the relationship stated in equation \eqref{eq:expression_to_be_proven} is provided in Table \ref{table:example_Pascal_triangle_rule}.
Expression \eqref{eq:expression_to_be_proven} might be already known. However I could not find it elsewhere.

\begin{table}[h]
% The example is based on the example from http://forum.mackichan.com/node/274
\vspace{-0.5cm}
\begin{tabular}{rccccccccccccccccc}
$N=0$: & & & & & & & & & 1 & & & & \\
$N=1$: & & & & & & & & 1 & & 1 & & & \\
$N=2$: & & & & & & & 1 & & 2 & & 1 & & \\
$N=3$: & & & & & & 1 & & 3 & & 3 & & \textbf{1} & \\
$N=4$: & & & & & 1 & & 4 & & 6 & & \textbf{4} & & 1 & \\
$N=5$: & & & & 1 & & 5 & & 10 & & \textbf{10} & & 5 & & 1 &\\
$N=6$: & & & 1 & & 6 & & 15 & & \textbf{20} & & 15 & & 6 & & 1 & \\
$N=7$: & & 1 & & 7 & & 21 & & \textbf{35} & & 35 & & 21 & & 7 & & 1 & \\
$N=8$: & 1 & & 8 & & 28 & & 56 & & $\textbf{70}$ & & \textbf{56} & & \textbf{28} & & \textbf{8} & & \textbf{1} \\
\end{tabular}
\caption{Example of the relationship \eqref{eq:expression_to_be_proven} between the elements in a row and the diagonal whose leftmost elements are adjacent in a way demonstrated. The elements used in the computation are shown in bold. Here $N = 8$, $i = 4$. The computation is as follows: $-1 \cdot 70 + 4 \cdot 56 - 10 \cdot  28 + 20 \cdot  8 - 35 \cdot 1 = -1$.}\label{table:example_Pascal_triangle_rule}
\end{table}

\begin{proof}
\par Induction\footnote{Proposition \ref{proposition_eq:expression_to_be_proven} proven here by induction can also be proven by looking for the coefficient of $x^{N-i}$ in the product $(1+x)^{-i} \cdot (1+x)^N$ as suggested by Ron Adin.} is used to derive the formula \eqref{eq:expression_to_be_proven}. For the row with $N = 3$ in Table \ref{table:example_Pascal_triangle} the equality \eqref{eq:expression_to_be_proven} is true for any value of $i \leq 3$. Assume the equality is true for some $N$ and for an arbitrary $i \leq N$ at row $N$. Further rewrite \eqref{eq:expression_to_be_proven} for the row $N + 1$ and use \eqref{eq:define_alpha1} for $\alpha_{N+1,\, k}$:
\begin{eqnarray}
  \nonumber && \sum_{k = i}^{N + 1} (-1)^{i - k + 1} \alpha_{k - 1,\, i - 1} \cdot \alpha_{N + 1,\, k} =  \sum_{k = i}^{N + 1} (-1)^{i - k + 1} \alpha_{k - 1,\, i - 1} \cdot \left( \alpha_{N,\, k - 1} + \alpha_{N,\, k} \right) =  \\
  \nonumber  &&  \sum_{k = i}^{N + 1} (-1)^{i - k + 1} \alpha_{k - 1,\, i - 1} \cdot  \alpha_{N,\, k - 1} + \left( \sum_{k = i}^{N} (-1)^{i - k + 1} \alpha_{k - 1,\, i - 1} \cdot  \alpha_{N,\, k} + \right. \\
  \nonumber &&  \left. (-1) ^{i - (N + 1) + 1} \alpha_{N + 1 - 1,\, i - 1} \cdot \alpha_{N,\, N + 1} \right)
\end{eqnarray}
From \eqref{eq:boundary_conditions_alpha} $\alpha_{N,\, N + 1} = 0$ and from \eqref{eq:expression_to_be_proven} and the assumption of induction $\sum_{k = i}^{N} (-1)^{i - k + 1} \alpha_{k - 1,\, i - 1} \cdot  \alpha_{N,\, k} = -1$. Therefore, and using again \eqref{eq:define_alpha1} for $\alpha_{k-1,\, i-1}$,
\begin{eqnarray}
    \nonumber && \sum_{k = i}^{N + 1} (-1)^{i - k + 1} \alpha_{k - 1,\, i - 1} \cdot \alpha_{N + 1,\, k} =  \sum_{k = i}^{N + 1} (-1)^{i - k + 1} \alpha_{k - 1,\, i - 1} \cdot  \alpha_{N,\, k - 1} - 1 + 0 = \\
    \nonumber &&  \sum_{k = i}^{N + 1} (-1)^{i - k + 1} \left(\alpha_{k - 2,\, i - 2} + \alpha_{k - 2,\, i - 1} \right) \cdot \alpha_{N,\, k - 1} - 1 = \\
    \nonumber &&  \sum_{k = i}^{N + 1} (-1)^{i  - k + 1} \alpha_{N,\, k - 1} \cdot \alpha_{k - 2,\, i - 1} + \sum_{k = i - 1}^{N} (-1)^{(i - 1) - k + 1} \alpha_{N,\, k} \cdot \alpha_{k - 1,\, (i - 1) - 1} - 1 = \\
    \nonumber &&  (-1)^{i - i + 1}  \alpha_{N,\, i - 1} \cdot \alpha_{i - 2,\, i - 1} + \sum_{k = i + 1}^{N + 1} (-1)^{i - k + 1} \alpha_{N,\, k - 1} \cdot \alpha_{k - 2,\, i - 1} - 1 - 1\; .
\end{eqnarray}
Now note that $\alpha_{i - 2,\, i - 1} = 0$ from \eqref{eq:boundary_conditions_alpha} and use the substitution $j = k - 1$ to get:
\begin{eqnarray}
    \nonumber && \sum_{k = i}^{N + 1} (-1)^{i - k + 1} \alpha_{k - 1,\, i - 1} \cdot \alpha_{N + 1,\, k} = -2 + (-1) \cdot \sum_{j = i}^{N} (-1)^{i - j + 1} \alpha_{N,\, j} \cdot \alpha_{j - 1,\, i - 1} = \\
    \nonumber && -2 + 1 = -1
 \end{eqnarray}
\end{proof}

\subsection*{Application of the derived property \eqref{eq:expression_to_be_proven} to the Euler-Poisson equation}
\par An optimization problem \eqref{solution_path_minimization_problem1}
is under consideration.
\par The problem of finding vector functions for which the speed of accumulating the parameter $\sigma(t)$ is constant and is constrained at the boundaries is stated in the manuscript. Derivation of the differential equation satisfied by the solutions of the optimization problem applies Euler-Poisson equation and leads to the following expression from \eqref{eq:I_n_x_expressed_x_v}:
\begin{equation}\label{eq:result_of_euler_poisson}
 (2n - 1) \left[x^{(n)} \right]^2 + 2 \sum_{k = 2}^{n}(-1)^{k - 1} \binom{n}{k}\cdot \frac{d^{k - 1}}{d\sigma^{k - 1}}\left[x^{(n)} \cdot x^{(n-(k - 1))} \right]\, ,
\end{equation}
where
$$x^{(n)} \equiv \frac{d^{n}}{d\sigma^{n}} x(\sigma)\, .$$
Here I show how property \eqref{eq:expression_to_be_proven} helps to rewrite the expression \eqref{eq:result_of_euler_poisson} in a simpler form:
\begin{eqnarray}
\nonumber &&(2n - 1) \left[x^{(n)} \right]^2 + 2 \sum_{k = 2}^{n}(-1)^{k - 1} \binom{n}{k}\cdot \frac{d^{k - 1}}{d\sigma^{k - 1}}\left[x^{(n)} \cdot x^{(n-(k - 1))} \right] = \\
\nonumber && \left[x^{(n)}\right]^2 - 2  x^{(n-1)} x^{(n+1)} + 2 x^{(n-2)} x^{(n+2)}+ \ldots + (-1)^{n-1} \cdot 2 x' x^{(2n - 1)} = \\
\label{eq_simpler_form} && \left[x^{(n)}\right]^2 + 2 \sum_{i = 1}^{n - 1} (-1)^{i}  x^{(n - i)} x^{(n + i)}\; .
\end{eqnarray}
Interestingly, differentiation of \eqref{eq_simpler_form} leads to simply a constant multiplied by a single product of the 1st and $2n$-th order derivatives of $x$ with respect to $\sigma$:
\begin{equation*}
    \frac{d}{d\sigma}\left\{ \left[x^{(n)}\right]^2 + 2 \sum_{i = 1}^{n - 1} (-1)^{i}  x^{(n - i)} x^{(n + i)} \right\} =  2 \cdot (-1)^{n - 1} \cdot x' x^{(2 n)}\, .
\end{equation*}

\subsection*{Derivation of the expression \eqref{eq_simpler_form}}
Applying Leibnitz rule to \eqref{eq:result_of_euler_poisson}
\begin{eqnarray}
\nonumber && (2n - 1) \left[x^{(n)} \right]^2 + 2 \sum_{k = 2}^{n}(-1)^{k - 1} \binom{n}{k}\cdot \frac{d^{k - 1}}{d\sigma^{k - 1}}\left[x^{(n)} \cdot x^{(n-(k - 1))} \right] = \\
\nonumber && (2n - 1) \left[x^{(n)} \right]^2 + 2 \sum_{k = 2}^{n}(-1)^{k - 1} \binom{n}{k}\cdot \sum_{j = 0}^{k - 1}\binom{k - 1}{j}x^{(n + j)} x^{(n - (k - 1) + k - 1- j)} = \\
\nonumber && (2n - 1) \left[x^{(n)} \right]^2 + 2 \sum_{j = 0}^{n - 1} x^{(n + j)} x^{(n - j)} \left[ \sum_{k = j + 1}^{n}  (-1)^{k - 1} \binom{k - 1}{j} \cdot \binom{n}{k} \right] - 2 n \left[x^{(n)} \right]^2\, .
\end{eqnarray}
Substituting $i = j + 1$:
\begin{eqnarray}
\nonumber && (2n - 1) \left[x^{(n)} \right]^2 + 2 \sum_{k = 2}^{n}(-1)^{k - 1} \binom{n}{k}\cdot \frac{d^{k - 1}}{d\sigma^{k - 1}}\left[x^{(n)} \cdot x^{(n-(k - 1))} \right] = \\
\nonumber && - \left[x^{(n)} \right]^2 + 2 \sum_{i = 1}^{n} x^{(n + (i - 1))} x^{(n - (i - 1))} \left[ \sum_{k = i}^{n}  (-1)^{k - 1} \binom{k - 1}{i - 1} \cdot \binom{n}{k} \right] = \\
\nonumber  &&  - \left[x^{(n)} \right]^2 + 2 \sum_{i = 1}^{n} x^{(n + (i - 1))} x^{(n - (i - 1))} \left[ \sum_{k = i}^{n}  (-1)^{- k + 1} \binom{k - 1}{i - 1} \cdot \binom{n}{k} \right] = \\
\label{eq:mark_substitution}  &&  - \left[x^{(n)} \right]^2 + 2 \sum_{i = 1}^{n} x^{(n + (i - 1))} x^{(n - (i - 1))} \left[ \sum_{k = i}^{n} (-1) ^{i} \cdot (-1)^{i - k + 1} \binom{k - 1}{i - 1} \cdot \binom{n}{k} \right] = \\
\nonumber  &&  - \left[x^{(n)} \right]^2 + 2 \sum_{i = 1}^{n} (-1)^{i + 1} x^{(n + (i - 1))} x^{(n - (i - 1))} = \left[x^{(n)} \right]^2 + 2 \sum_{i = 2}^{n} (-1)^{i + 1} x^{(n + (i - 1))} x^{(n - (i - 1))} = \\
\nonumber  &&  \left[x^{(n)} \right]^2 + 2 \sum_{j = 1}^{n - 1} (-1)^{j} x^{(n + j )} x^{(n - j)}\; .
\end{eqnarray}
Property \eqref{eq:expression_to_be_proven} was used for the expression in brackets in \eqref{eq:mark_substitution}.
$\;\;\;\;\;\;\;\;\;\Box$

\vspace{1cm}

\bibliographystyle{amsplain}
%\bibliography{XBib}

%\nolinenumbers
%\bibliography{XBib}
%\bibliographystyle{vancouver} %{bmc-mathphys} % Style BST file (bmc-mathphys, vancouver, spbasic).
\bibliography{xbib}

%\printindex
\end{document}